\RequirePackage{amsmath}
\documentclass[runningheads]{llncs}

\newif\ifwithappendix
\newif\ifwithnotes
\newif\ifwithcomments
\newif\ifdraft
\newif\iftechnicalreport

\draftfalse

\technicalreporttrue

\ifdraft
    \withnotesfalse
    \withcommentstrue

    \iftechnicalreport
        \withappendixtrue
    \fi
\else
    \withnotesfalse
    \withcommentsfalse

    \iftechnicalreport
        \withappendixtrue
    \fi
\fi

\usepackage[T1]{fontenc}
\usepackage{environ}
\usepackage{graphicx}
\usepackage{mdframed}
\usepackage{amsfonts}
\usepackage{amssymb}
\usepackage{thm-restate}
\usepackage{stmaryrd}
\usepackage{xcolor}
\usepackage[hidelinks]{hyperref}
\usepackage{mathtools}
\usepackage{algorithm}
\usepackage{algorithmicx}
\usepackage[indLines]{algpseudocodex}
\usepackage{comment}
\usepackage{microtype}
\usepackage[caption=false]{subfig}
\usepackage{bbding}
\usepackage{enumitem}
\usepackage[outline]{contour}
\usepackage[capitalise, nameinlink]{cleveref}

\contourlength{1pt}

\newenvironment{proofs}{%
  \proof}{\endproof}

\newenvironment{myproof}{\proof}{\strut\qed\endproof}
\newenvironment{myexample}{\example}{\hfill$\blacksquare$\endexample}

\crefname{section}{Sect.}{Sects.}
\Crefname{section}{Sect.}{Sects.}
\crefname{definition}{Def.}{Defs.}
\Crefname{definition}{Def.}{Defs.}
\crefname{theorem}{Thm.}{Thms.}
\Crefname{theorem}{Thm.}{Thms.}
\crefname{lemma}{Lem.}{Lems.}
\Crefname{lemma}{Lem.}{Lems.}
\crefname{corollary}{Cor.}{Cors.}
\Crefname{corollary}{Cor.}{Cors.}
\crefname{algorithm}{Alg.}{Algs.}
\Crefname{algorithm}{Alg.}{Algs.}
\crefname{figure}{Fig.}{Figs.}
\Crefname{figure}{Fig.}{Figs.}
\crefname{example}{Ex.}{Exs.}
\Crefname{example}{Ex.}{Exs.}
\crefname{table}{Tbl.}{Tbls.}
\Crefname{table}{Tbl.}{Tbls.}

\newcommand{\lfpp}[0]{\mu}
\newcommand{\gfpp}[0]{\nu}
\DeclarePairedDelimiter\pars{\lparen}{\rparen}
\DeclarePairedDelimiter\tuple{\langle}{\rangle}
\DeclarePairedDelimiter\abs{\lvert}{\rvert}
\DeclarePairedDelimiter\Set{\{}{\}}
\DeclarePairedDelimiter\bracs[]
\DeclarePairedDelimiter\sem\llbracket\rrbracket

\newcommand{\powerset}[1]{\mathcal P\pars*{#1}}
\newcommand{\defeq}{\vcentcolon=}

\newcommand{\defequiv}{\overset{\mathit{def}}{\iff}}
\newcommand{\assign}[0]{\leftarrow}

\newcommand{\mdpc}[1]{#1}
\newcommand{\mdp}[0]{{\mdpc M}}

\newcommand{\graph}[0]{\mathcal{G}}
\newcommand{\verts}[0]{V}
\newcommand{\edges}[0]{E}
\newcommand{\mdpgraph}[0]{{\mdp}}

\newcommand{\Reach}[0]{\mathrm{Reach}}

\newcommand{\Post}[0]{\mathrm{Post}}
\newcommand{\Pre}[0]{\mathrm{Pre}}

\newcommand{\shortcut}{{\graph_\sd}}

\newcommand{\omdpc}[1]{\mathcal{#1}}

\newcommand{\io}[0]{\mathsf{IO}}
\newcommand{\seqcomp}[0]{\fatsemi}
\newcommand{\sumcomp}[0]{\oplus}
\newcommand{\sdc}[1]{\mathbb{#1}}
\newcommand{\sd}[0]{{\sdc D}}
\newcommand{\Path}[0]{\pi}

\newcommand{\sol}{\mathrm{S}}

\newcommand{\solarg}[1]{\sol_{#1}}
\newcommand{\conf}[1]{\mathcal{#1}}

\newcommand{\Exits}[2]{\mathrm{Exits}_{#1}\pars*{#2}}
\newcommand{\CExits}[3]{\mathrm{Exits}^{#1}_{#2}\pars*{#3}}

\newcommand{\strat}{\sigma}
\newcommand{\strats}{\Sigma}
\newcommand{\nolosestrats}{\strats_{\mathit{NL}}}

\newcommand{\constMDP}[1]{\mathsf{c}_{#1}}
\newcommand{\typemdp}[1]{\mathrm{arity}(#1)}
\newcommand{\en}{\mathsf{I}}
\newcommand{\ex}{\mathsf{O}}
\newcommand{\enlet}{\mathsf{i}}
\newcommand{\exlet}{\mathsf{o}}

\newcommand{\enarg}[1]{\enlet_{#1}}

\newcommand{\exarg}[1]{\exlet_{#1}}

\newcommand{\componentsent}{\mathrm{CPI}}
\newcommand{\solstrat}{\mathrm{Effect}}

\newcommand{\colorpar}[3]{\colorbox{#1}{\parbox{#2}{#3}}}
\newcommand{\marginremark}[3]{\marginpar{\colorpar{#2}{\linewidth}{\color{#1}#3}}}

\ifwithcomments
    \newcommand{\marck}[1]{\textcolor{green!70!black}{\textbf{MvdV}: #1}}
    \newcommand{\mvdv}[1]{\marginremark{white}{green!70!black}{\scriptsize{[MvdV]~ #1}}}
    \newcommand{\kazuki}[1]{\textcolor{blue!70!black}{\textbf{KW}: #1}}
    \newcommand{\kw}[1]{\marginremark{white}{blue!70!black}{\scriptsize{[KW]~ #1}}}
    \newcommand{\sebastian}[1]{\textcolor{purple!70!black}{\textbf{Sebastian}: #1}}
    \newcommand{\sj}[1]{\marginremark{white}{purple!70!black}{\scriptsize{[SJ]~ #1}}}
    \newcommand{\ichiro}[1]{\textcolor{cyan!70!black}{\textbf{Ichiro}: #1}}
    \newcommand{\ih}[1]{\marginremark{white}{cyan!70!black}{\scriptsize{[IH]~ #1}}}
\else
    \newcommand{\marck}[1]{}
    \newcommand{\mvdv}[1]{}
    \newcommand{\kazuki}[1]{}
    \newcommand{\kw}[1]{}
    \newcommand{\sebastian}[1]{}
    \newcommand{\sj}[1]{}
    \newcommand{\ichiro}[1]{}
    \newcommand{\ih}[1]{}
\fi

\newcommand{\nat}{{\mathbb N}}

\newcommand{\tr}{\mathrm{tr}}

\newcommand{\midvert}{\;\middle\vert\;}

\newcommand{\buchiobj}[1]{\always\eventually #1}
\newcommand{\buchiverts}{B}
\newcommand{\win}[2]{\mathrm{Win}_{#1}^{#2}}

\newcommand{\always}{\square}
\newcommand{\eventually}{\lozenge}

\newcommand{\mappingfunc}[2]{\mathrm{Conn}^{#1}_{#2}}
\newcommand{\eqcomment}[1]{\shortintertext{\hfill\scriptsize\itshape $(*$ {#1} $*)$ }}

\makeatletter
\newcommand{\svdots}{%
  \vbox{\fontsize{\sf@size}{\sf@size pt}\linespread{0.7}\selectfont
    \kern0.5\baselineskip
    \hbox{.}\hbox{.}\hbox{.}%
    \kern0.1\baselineskip
  }%
}
\makeatother

\newcommand{\refinementalgtext}{StratRef}

\newcommand{\canreachtext}{CanReach}
\newcommand{\canreach}{\textsc{\canreachtext}}

\newcommand\scalemath[2]{\scalebox{#1}{\mbox{\ensuremath{\displaystyle #2}}}}

\usepackage{tikz}
\usetikzlibrary{
    positioning,
    patterns,
    patterns.meta,
    arrows.meta,
    calc,
    decorations.pathmorphing,
    decorations.pathreplacing,
    shapes,
} 

\tikzset{
	omdp/.style={ draw=black, thick, rounded corners=5, },
	omdpA/.style={ fill=red!50, },
	omdpB/.style={ fill=green!50, },
	omdpC/.style={ fill=blue!50, },
	omdpD/.style={ fill=yellow!50, },
    partial ellipse/.style args={#1:#2:#3}{
        insert path={+ (#1:#3) arc (#1:#2:#3)}
    },
	arr/.style={ thick, {Circle[open,fill=white]}-{Stealth} },
	arrr/.style={ thick, {Stealth}-{Circle[open,fill=white]} },
	oarr/.style={ draw=none, {Circle[open,fill=white]}- },
	oarrr/.style={ draw=none, -{Circle[open,fill=white]} },
	Arr/.style={ thick, {-{Stealth[length=1.5mm]}} },
}

\tikzset{invclip/.style={clip,insert path={{[reset cm]
      (-16383.99999pt,-16383.99999pt) rectangle (16383.99999pt,16383.99999pt)
    }}}}

\iftechnicalreport
\renewcommand{\orcidID}[1]{}
\fi

\title{Compositional Verification of Almost-Sure~Büchi~Objectives~in~MDPs%
\thanks{%
I.H. was supported by ERATO HASUO Metamathematics for Systems Design Project (JST, No.\ JPMJER1603).
I.H. and M.V. were supported by the ASPIRE grant (JST, No.\ JPMJAP2301).
K.W. was supported by ACT-X grant (JST, No.\ JPMJAX23CU).
}%
}
\author{Marck van der Vegt\inst{1}{\tiny\orcidID{0000-0003-2451-5466}} \and Kazuki Watanabe\inst{2}{\tiny\orcidID{0000-0002-4167-3370}} \and\\ Ichiro Hasuo\inst{2}{\tiny\orcidID{0000-0002-8300-4650}} \and Sebastian Junges\inst{1}{\tiny\orcidID{0000-0003-0978-8466}}}
\authorrunning{M. van der Vegt et al.}
\institute{
    Radboud University, The Netherlands\\
    \email{\{marck.vandervegt,sebastian.junges\}@ru.nl}
    \and
    National Institute of Informatics, Japan\\
    \email{\{kazukiwatanabe,hasuo\}@nii.ac.jp}
}
\begin{document}

\maketitle

\begin{abstract}
This paper studies the verification of almost-sure Büchi objectives in MDPs with a known, compositional structure based on string diagrams.
In particular, we ask whether there is a strategy that ensures that a Büchi objective is almost-surely satisfied. 
We first show that proper exit sets --- the sets of exits that can be reached within a component without losing locally --- together with the reachability of a Büchi state are a sufficient and necessary statistic for the compositional verification of almost-sure Büchi objectives.
The number of proper exit sets may grow exponentially in the number of exits.
We define two algorithms:
(1) A straightforward bottom-up algorithm that computes this statistic in a recursive manner to obtain the verification result of the entire string diagram and
(2) a polynomial-time iterative algorithm which avoids computing all proper exit sets by performing iterative strategy refinement.
\end{abstract}

\section{Introduction}
Markov decision processes (MDPs) are a ubiquitous model to describe systems with uncertain action outcomes. Their efficient verification is hindered by typical state space explosion problems, see e.g., \cite{DBLP:journals/corr/abs-2405-13583,DBLP:conf/tacas/HartmannsJQW23}.
To overcome this state space explosion, a promising research direction to overcome scalability concerns is to explicitly capture and utilize the (compositional) structure of the MDP~\cite{ClarkeLM89,HauskrechtMKDB98,NearyVCT22}.

Concretely, we consider sequentially composed MDPs (potentially with loops), that is, we consider MDPs whose state space is partitioned into \emph{components}.
Specifically, we study the setting in which these components and their interaction is specified as string diagrams~\cite{DBLP:conf/cav/WatanabeEAH23}:
That is, components are MDPs with entrance and exit states and their composition is defined by a number of operations that glue these components together to obtain the \emph{monolithic MDP} (see e.g., \cref{fig:example_omdp}).
The key concept in compositional verification is then to verify objectives without constructing the monolithic MDP, but rather, by (iteratively or exhaustively) analyzing the components and combining the result of their analysis.
This has been done for quantitative reachability in~\cite{DBLP:conf/cav/WatanabeEAH23,DBLP:conf/cav/WatanabeVJH24}, but not for almost-sure Büchi objectives.
Although the analysis of Büchi objectives on a monolithic MDP can be done in polynomial time in the number of states of the MDP, the size of the MDPs is often prohibitively large.
Compositional approaches avoid the necessity to handle the monolithic MDP.

The key question this paper answers affirmatively is whether such a compositional approach can work for almost-sure Büchi objectives.
We study MDPs with a set of marked states, called the \emph{Büchi states}, and algorithmically answer whether given an entrance it is possible to \emph{win} the almost-sure Büchi objective, that is, can we visit a Büchi state infinitely often with probability one.

In this paper, we first discuss what information is necessary to sufficiently describe the behavior of the components \emph{exhaustively}.
We call this information the \emph{solution}. %
The solution should include for each entrance the set of exits that can be reached using strategies that we call \emph{no-lose}, that is, a strategy which avoid surely violating the Büchi objective in the component.
We call the set of reachable exits the \emph{proper exit set}.
Intuitively, after playing a no-lose strategy, we can reach a proper exit which is connected to an entrance of another component, after which we play another no-lose strategy.
However, constantly playing no-lose strategies is not sufficient to satisfy the Büchi objective.
In the solution, beside the proper exit set we must therefore also include whether there is a positive probability of reaching a Büchi state.
If for a given entrance, there exists a no-lose strategy that can reach a Büchi state, then reaching this entrance infinitely often implies satisfaction of the Büchi objective.
We call the combination of the proper exit set and the positive reachability of a Büchi state the \emph{effect} of the strategy and it is \emph{the} essential information for our \emph{compositional solution}.

The characterization of the problem in terms of a compositional solution gives rise to a natural (exhaustive) algorithm in the spirit of \cite{DBLP:conf/cav/WatanabeEAH23,DBLP:conf/tacas/WatanabeVHRJ24}:
Consider the abstract syntax tree of the string diagram, first analyze the leafs of this tree exhaustively by computing its solution, and then compose the analysis results.
The exhaustive analysis is exponential in the number of exits, but polynomial in the number of components.
Compared to the quantitative reachability case~\cite{DBLP:conf/cav/WatanabeEAH23,DBLP:conf/tacas/WatanabeVHRJ24}, this is already good news, because in the quantitative case, the computation can be superpolynomial even with a constant number of exits per component.

Our second (iterative) algorithm avoids computing all possible effects by taking the context of each component into account, inspired by a compositional formulation of value iteration in \cite{DBLP:conf/cav/WatanabeVJH24}.
In contrast to existing work, this algorithm can limit the amount of iterations necessary by determining that certain entrances may not be reached by any winning strategy. %
In particular, we start with a strategy that reaches as many entrances as possible, and refine this strategy by pruning entrances from which we cannot avoid losing.
Our iterative algorithm mimics the \emph{classical Büchi algorithm}~\cite{DBLP:conf/lics/AlfaroH00,DBLP:phd/us/Alfaro97,DBLP:conf/csl/ChatterjeeJH03} on the level of components.
We show that the sets of possible effects computed by the exhaustive algorithm exhibit a join semilattice structure, which the iterative algorithm uses to prevent computing all effects, yielding a polynomial time algorithm.

In conclusion,
the contributions of this work are a characterization of the necessary and sufficient information about the behavior of a component in order to prove satisfaction of a Büchi objective globally,
an exponential-time bottom-up algorithm that computes this information on individual components and composes them according to the structure of the string diagram,
and an iterative approach that avoids the exponential influence of the number of exits. 

\section{Preliminaries}\label{sec:prelims}
For $n \in \nat$, $[n]$ denotes the set $\{1,2, \dots, n \}$.
We denote the disjoint union by $\uplus$.
We use $\gfpp$ and $\lfpp$ to denote the greatest and least fixpoint operators, respectively.

A (directed) \emph{graph} $\graph = \tuple{\verts, \edges}$ is a tuple with finitely many vertices $\verts$ and edges $\edges \subseteq \verts \times \verts$.
\emph{The successors} of $x \in \verts$ are $\Post(x) \defeq \{ y \in \verts \mid \tuple{x, y} \in \edges \}$. Similarly, $\Pre(y) \defeq \{ x \in \verts \mid \tuple{x, y} \in \edges \}$ are the \emph{predecessors} of $y \in \verts$.
We lift $\Post$ and $\Pre$ to sets: $\Post(X) \defeq \bigcup_{x \in X} \Post(x)$, $\Pre(X) \defeq \bigcup_{x \in X} \Pre(x)$.

We are interested in \emph{almost-sure acceptance} in \emph{Markov decision processes}~\cite{DBLP:books/daglib/0020348,DBLP:conf/sfm/ForejtKNP11} (MDPs). 
Then, the amplitude of probabilities is irrelevant---a non-zero probability, however small it is, can break almost-sure acceptance. Therefore, we study the following abstraction of MDPs (see, e.g.~\cite{DBLP:conf/soda/ChatterjeeH11}), which preserves all results.

\begin{definition}[MDP graph and strategy]\label{def:mdp_and_strategy}
An \emph{MDP graph} $\mdpgraph \defeq \tuple{\verts_1, \verts_P, \edges}$ is a graph $\tuple{\verts, \edges}$ where $\verts \defeq \verts_1 \uplus \verts_P$. $\verts_1$ are the \emph{player-1 vertices} and $\verts_P$ are the \emph{probabilistic vertices}.
A \emph{strategy} $\strat\colon \verts_1 \to \powerset{\verts_P}$ maps player-1 vertices to successor vertices, such that for all vertices $v \in \verts_1$: $\strat(v) \subseteq \Post(v)$.
\end{definition}
In the following we simply refer to MDP graphs as MDPs.
We denote the pointwise union of strategies $\strat_1, \strat_2$ by $\strat_1 \cup \strat_2$, with $(\strat_1 \cup \strat_2)(v) \defeq \strat_1(v) \cup \strat_2(v)$ for all $v \in \verts_1$.
We note that our strategies are memoryless randomized strategies. These are sufficient for satisfying almost-sure Büchi objectives~\cite{DBLP:conf/lics/AlfaroH00}.
We assume w.l.o.g. that $\verts_1$ and $\verts_P$ vertices \emph{alternate}, i.e., $\verts_1$ vertices are only reachable from $\verts_P$ vertices and vice versa.

\begin{definition}[reachable vertices]\label{def:reachable_vertices}
Let $\Post_\strat(x)$ be equal to $\strat(x)$ if $x \in \verts_1$ and $\Post(x)$ otherwise.
We define the \emph{vertices reachable from $X \subseteq \verts$ under $\strat$} as 
$\Reach_\strat(X) \defeq \bigcup_{n=0}^\infty X_n$, where $X_0 \defeq X$ and $X_{n+1} \defeq X_n \cup \Post_\strat(X_n)$.
\end{definition}
The general notion of Büchi acceptance is over \emph{infinite traces}.  
When one restricts to almost-sure acceptance (as we do), there is an alternative characterization~\cite[Thm.~10.29, Cor.~10.30]{DBLP:books/daglib/0020348}, namely that a vertex $i$ is \emph{almost-sure accepting} iff there exists a strategy $\strat$ such that for every vertex $v$ reachable from $i$ under $\strat$, a Büchi vertex $v'$ is reachable from $v$ under $\strat$.

This characterization justifies the following definition.
\begin{definition}[Büchi objective]\label{def:buchi}
A strategy $\strat$ \emph{satisfies the Büchi objective} $\buchiverts \subseteq \verts_P$ for vertex $i \in \verts$ if $\Reach_\strat(v) \cap \buchiverts \ne \emptyset$ for all $v$ in $\Reach_\strat(i)$.
We denote $i, \strat \models \buchiobj\buchiverts$ or simply $i \models \buchiobj\buchiverts$ if such a strategy exists.  
We denote the \emph{winning region} by $\win{\buchiobj\buchiverts}{} \defeq \{ v \in \verts_1 \mid v \models \buchiobj\buchiverts \}$.
\end{definition}
For ease of presentation, we assume w.l.o.g. that all Büchi vertices are probabilistic vertices.
We briefly recall the \emph{classical Büchi algorithm}~\cite{DBLP:conf/lics/AlfaroH00,DBLP:phd/us/Alfaro97,DBLP:conf/csl/ChatterjeeJH03}.

\begin{definition}[Büchi operator]\label{def:buchi_operator}
Let $F(X, Y) \defeq \{ v \in \verts_1 \mid \exists v' \in \Post(v)\colon\\ \Post(v') \subseteq Y \text{ and } v' \in \buchiverts \vee v' \in \Pre(X) \}$, the \emph{Büchi operator}.
\end{definition}%
We can compute the winning region $\win{\buchiobj\buchiverts}{}$ using the Büchi operator:
\begin{restatable}{lemma}{lembuchi}\label{lem:buchi}
$\nu Y.\, \mu X.\, F(X, Y) = \win{\buchiobj\buchiverts}{}$
\end{restatable}\noindent
For a fixed $Y$, a vertex $v$ is in the `inner fixpoint' $\mu X.\, F(X, Y)$ iff it can reach a Büchi vertex while staying in $Y$.
Then, a vertex $v$ is in the `outer fixpoint' $\nu Y.\, \mu X.\, F(X, Y)$ iff all vertices reachable from $v$ can reach a Büchi vertex, which is exactly the winning region for Büchi acceptance.
The classical Büchi algorithm computes the above fixpoint using Kleene iteration. %
Initially, we fix $Y$ to $\verts_1$ and compute the least fixpoint of $F$ by iteratively applying $F$ to $\emptyset$, i.e., we compute the sequence $F(\emptyset, Y), F(F(\emptyset, Y), Y), \cdots$.
The obtained fixpoint is the new value of $Y$.
The process is repeated until we obtain a fixpoint of $Y$, in which case $Y$ is equal to the winning region $\win{\buchiobj\buchiverts}{}$. %

\ifwithappendix
All proofs are given in the appendix.
\else
All proofs are given in the appendix of the technical report~\cite{technicalreport}.
\fi

\subsection{Compositional MDPs}\label{sec:compositional_mecs}
We define \emph{rightward open MDP} which have \emph{open ends} via which they compose. 

\begin{definition}[roMDP]
A \emph{rightward open MDP} $\omdpc A = \tuple{\mdpgraph, \io}$ is a pair consisting of an MDP $\mdpgraph$ with \emph{open ends} $\io = \tuple{\en, \ex}$, where
$\en, \ex \subseteq \verts_1$ are pairwise disjoint and totally ordered.
The vertices in $\en$ and $\ex$ are the \emph{entrances} and \emph{exits} respectively.
We require that $\Pre(\en) = \Post(\ex) = \emptyset$. %
\end{definition}
We denote the \emph{arity} of $\omdpc A$ as $\typemdp{\omdpc A}\colon m \to n$ meaning that $\abs I = m$ and $\abs O = n$.
In the following, we will fix arbitrary roMDPs $\omdpc A$ and $\omdpc B$.
We often use superscript to refer to elements of an roMDP $\omdpc A$, e.g., $\en^{\omdpc A}$ denotes the entrances of $\omdpc A$.

\begin{figure}[t]
    \centering
    \input{tikz/example1}
    \caption{$\omdpc C \defeq \tr\bigl(\tr(\omdpc A \seqcomp \omdpc B)\bigr)$ and its operational semantics.}
    \label{fig:example_omdp}
\end{figure}

\begin{definition}[string diagram]
A \emph{string diagram $\sd$ of roMDPs} is a term adhering to $\sd \defeq \mathsf{c}_{\omdpc A} \mid \sd \seqcomp \sd \mid \sd \oplus \sd \mid \tr(\sd)$, where $\constMDP{\omdpc A}$ designates an roMDP $\omdpc A$. 
\end{definition}
In the definition of a string diagram, $\seqcomp$ denotes \emph{sequential composition}, $\sumcomp$ denotes \emph{sum composition} and $\tr$ denotes \emph{trace}.
In the following we fix $\sd$ to be an arbitrary string diagram of roMDPs.
\Cref{fig:example_omdp} demonstrates the graphical intuition behind trace and sequential composition.
In the figure, the clouds represent a large number of vertices and transitions. 
We now define these operations. 

\begin{definition}[$\omdpc A \seqcomp \omdpc B$]\label{def:seqROMDP}
Let $\omdpc A$, $\omdpc B$ be roMDPs such that $\typemdp{\omdpc A}\colon m \to l$ and $\typemdp{\omdpc B}\colon l \to n$.
Their \emph{sequential composition} $\omdpc A \seqcomp \omdpc B$ is the roMDP $\tuple{\mdpgraph, \io}$ where
$\io \defeq \tuple{\en^{\omdpc A}, \ex^{\omdpc B}}$,
$\mdpgraph \defeq \tuple*{(\verts_1^{\omdpc A} \uplus \verts_1^{\omdpc B}) \setminus \ex^{\omdpc A}, \verts_P^{\omdpc A} \uplus \verts_P^{\omdpc B}, \edges}$ 
and $\edges$ satisfies the following: $\tuple{v, v'} \in \edges$ iff either:
(1) $\tuple{v, v'} \in \edges^{\omdpc A} \wedge v' \in \verts^{\omdpc A} \setminus \ex^{\omdpc A}$,
(2) $\tuple{v, \exarg i^{\omdpc A}} \in \edges^{\omdpc A} \wedge v' = \enarg i^{\omdpc B}$ for some $i \in [l]$,
(3) $\tuple{v, v'} \in \edges^{\omdpc B}$.
\end{definition}

\begin{definition}[$\omdpc A \sumcomp \omdpc B$]\label{def:sumROMDP}
\label{def:sumOMDP}
    The \emph{sum} $\omdpc A \oplus \omdpc B$ is the roMDP $\tuple{\mdpgraph, \io}$ where
    $\io \defeq \tuple{\en^{\omdpc A}\uplus \en^{\omdpc B}, \ex^{\omdpc A}\uplus \ex^{\omdpc B}}$,
    $\mdpgraph \defeq \tuple{\verts_1^{\omdpc A} \uplus \verts_1^{\omdpc B}, \verts_P^{\omdpc A} \uplus \verts_P^{\omdpc B}, \edges}$,
    and $\edges$ satisfies
    $\Post(v) = \Post^{\omdpc D}(v)$ for $\omdpc D \in \{\omdpc A, \omdpc B\}$ if $v \in \verts^{\omdpc D}$.
\end{definition}

\begin{definition}[$\tr(\omdpc A)$]\label{def:traceROMDP}
Let $\omdpc A$ be an roMDP such that $\typemdp{\omdpc A}\colon m+1 \to n+1$.
The \emph{trace} $\tr(\omdpc A)$ is the roMDP $\tuple{\mdpgraph, \io}$ where
$\io \defeq \tuple*{\en^{\omdpc A} \setminus \Set*{ \enarg{m+1}^{\omdpc A} }, \ex^{\omdpc A} \setminus \Set*{ \exarg{n+1}^{\omdpc A} }}$,
$\mdpgraph \defeq \tuple*{\verts_1^{\omdpc A}, \verts_P^{\omdpc A} \uplus \{*\}, \edges}$ and
$\edges \defeq \edges^{\omdpc A} \cup \Set*{ \tuple*{\exarg{n+1}^{\omdpc A}, *}, \tuple*{*, \enarg{m+1}^{\omdpc A}} }$.
\end{definition}
The trace operation introduces two edges that create a loop from the last exit to the last entrance. %
The `$*$' vertex is added to maintain vertex alternation.

\begin{definition}[operational semantics $\sem\sd$]\label{def:opsem}
The \emph{operational semantics} $\sem\sd$ is the roMDP inductively defined by \cref{def:seqROMDP,def:sumOMDP,def:traceROMDP}, with the base case $\sem{\constMDP{\omdpc A}}=\omdpc A$.
We assume that every string diagram $\sd$ has matching arities, w.r.t. the definitions of the compositional operators.
\end{definition}
We call $\sem\sd$ the \emph{monolithic roMDP}.
We will now formally state the main problem. %
\begin{mdframed}
\textbf{Main Problem Statement:} %
Given entrance $i \in \en^{\sem\sd}$, does $i \models^{\sem\sd} \buchiobj\buchiverts$?
\end{mdframed}

\section{Compositional Solution}\label{sec:bottom_up}
A natural idea, following~\cite{DBLP:conf/cav/WatanabeEAH23,DBLP:conf/tacas/WatanabeVHRJ24}, is to define a compositional algorithm which takes the string diagram, computes a solution for every leaf of the string diagram, and then composes these solutions into an answer on the complete string diagram. 
This section develops exactly such an approach. 
We first define a (local) solution:
This solution captures the necessary and sufficient information that needs to be extracted from an roMDP.
Next, we define the compositional operations on these solutions so that they can be composed into a final answer. 

We continue by introducing the necessary definitions for our compositional solution, which are \emph{no-lose strategies} and their \emph{effect}.
Intuitively, no-lose strategies represent \emph{local strategies} that can be part of a \emph{globally winning strategy} in $\sem\sd$.
\begin{definition}[no-lose strategy]\label{def:no_lose}
Let $\strat$ be a strategy in $\omdpc A$.
We say that $\strat$ is a \emph{no-lose strategy} from $i \in \verts$ if
for every $v \in \Reach_\strat(i)$ either $\Reach_\strat(v) \cap \ex \ne \emptyset$ or $v,\strat \models \buchiobj\buchiverts$.
Additionally we require that $\strat$ is \emph{$i$-local}, that is, $\strat(v) = \emptyset$ for $v \not\in \Reach_\strat(i)$.
We denote the set of no-lose strategies from $i$ by $\nolosestrats[i]$.
\end{definition}
The intuition of a no-lose strategy is as follows.
When playing such a strategy, we almost-surely either (1)~reach some exit or (2)~satisfy the Büchi condition.
A strategy that is not no-lose is \emph{losing}.
We will use the following characterization:
\begin{restatable}{lemma}{lemnoloseequivalentdef}\label{lem:no_lose_equivalent_def}
Strategy $\strat$ is no-lose from $i \in \verts$ iff for all $v$ in $\Reach_\strat(i)$ we have that $\Reach_\strat(v) \cap (\ex \cup B) \ne \emptyset$.
\end{restatable}\noindent
The lemma implies that no-lose strategies can be computed as strategies that satisfy the Büchi objective $\buchiobj{(\ex \cup \buchiverts)}$, using off-the-shelf algorithms (e.g.,~\cite{DBLP:conf/lics/AlfaroH00,DBLP:conf/csl/ChatterjeeJH03}). %

Towards determining whether a no-lose strategy $\strat$ is indeed part of a globally winning strategy, we consider its \emph{proper exit set} and \emph{effect}.
\begin{definition}[proper exit sets and effects]
Let $\strat$ be a no-lose strategy from vertex $v \in \verts$.
We denote the \emph{proper exit set} of $\strat$ as $\Exits\strat v \defeq \Reach_{\strat}(v) \cap \ex$.
We define the \emph{effect} of $\strat$ as
$\solstrat(v, \strat) \defeq \tuple*{\Exits{\strat}{v}, \Reach_{\strat}\pars*{v} \cap \buchiverts \ne \emptyset}.$
\end{definition}
We use $\tuple{T_1, b_1} \sqcup \tuple{T_2, b_2}$ to denote the pointwise join of effects, that is, $\tuple{T_1 \cup T_2, b_1 \vee b_2}$. 
By definition, $\tuple{\emptyset, \bot}$ cannot be the effect of a no-lose strategy $\strat$: it would imply that for all vertices $v$ we have that $\Reach_\strat(v) \cap (\ex \cup \buchiverts) = \emptyset$, contradicting that $\strat$ is no-lose.

\subsection{Local Solution}
We first introduce the solution for a given roMDP and then define the composition of such solutions.
Finally, we prove the compositionality in \cref{thm:semsol_compositional}.
Our solution contains all possible effects of no-lose strategies.

\begin{definition}[local solution]\label{def:semec}
The \emph{local solution} of roMDP $\omdpc A$ is the function $\solarg{\omdpc A}\colon \en^{\omdpc A} \to \powerset{\powerset{\ex^{\omdpc A}} \times \{\top, \bot\}}$ s.t.
$\solarg{\omdpc A}\pars*{\enarg i^{\omdpc A}} \defeq \bigl\{ \solstrat^{\omdpc A}\pars*{\enarg i^{\omdpc A}, \strat} \bigm| \strat \in \nolosestrats\bracs*{\enarg i^{\omdpc A}} \bigr\}$.
\end{definition}
We can replace every entrance and exit in the above definition by their index to get $\solarg{\omdpc A}\colon [m] \to \powerset{\powerset{[n]} \times \{\top, \bot\}}$, with the advantage that it no longer depends on the identity of the vertices.
For brevity, we (sometimes) implicitly convert between entrances/exits and their index, e.g., we denote $i$ for $\enarg i^{\omdpc A}$ and vice versa.
We define $\typemdp{\solarg{\omdpc A}}$ as $\typemdp{\omdpc A}$.
As a special case of local solutions, we can derive whether one can win within an roMDP without leaving it.
\begin{restatable}{lemma}{lemlocalsolutioncorrect}\label{lem:local_solution_correct}
For all $i \in \en^{\omdpc A}$, we have
$\tuple{\emptyset, \top} \in \solarg{\omdpc A}(i) \text{ iff } i \models^{\sem{\omdpc A}} \buchiobj\buchiverts$.
\end{restatable}\noindent
See \cref{fig:local_solution} for a graphical intuition of local solutions.

\subsection{Compositionality of the Solution}
The information in a local solution suffices to compositionally compute the solution of a string diagram.
We first state the theorem and then clarify the semantics of the operators $\seqcomp$, $\sumcomp$, and $\tr$ on the solutions in the following definitions.
\begin{theorem}[solution is compositional]\label{thm:semsol_compositional}
\begin{align*}
\solarg{\omdpc A \seqcomp \omdpc B} = \solarg{\omdpc A} \seqcomp \solarg{\omdpc B}, \qquad
\solarg{\omdpc A \sumcomp \omdpc B} = \solarg{\omdpc A} \sumcomp \solarg{\omdpc B}, \qquad \text{and} \qquad
\solarg{\tr(\omdpc A)} = \tr(\solarg{\omdpc A}).
\end{align*}
\end{theorem}

\subsubsection{Sequential Composition.}
As stated by \cref{thm:semsol_compositional}, we want the definition of $\solarg{\omdpc A} \seqcomp \solarg{\omdpc B}$ to be equal to $\solarg{\omdpc A \seqcomp \omdpc B}$.
We can combine the effects of $\solarg{\omdpc A}$ and $\solarg{\omdpc B}$ to obtain an effect in $\solarg{\omdpc A \seqcomp \omdpc B}$ by considering the underlying no-lose strategies.

We explain the intuition using \cref{fig:sequential_composition_solution}.
First, we choose a no-lose strategy $\strat$ in $\omdpc A$ for some entrance $i$, with the goal of constructing a no-lose strategy in $\omdpc A \seqcomp \omdpc B$.
Let $\tuple{\{x_1, \dots, x_N\}, b}$ be the effect of $\strat$ such that $N > 0$ (the case where $N=0$ is trivial, as it implies that the strategy is also a no-lose strategy in $\omdpc A \seqcomp \omdpc B$).
The strategy $\strat$ is not defined on the vertices of $\omdpc B$ and is therefore losing in $\omdpc A \seqcomp \omdpc B$.
Thus, to construct a no-lose strategy in $\omdpc A \seqcomp \omdpc B$, we need to ensure that the continuation in $\omdpc B$ is also no-lose.
If there exist no-lose strategies $\strat_1, \dots, \strat_N$ in $\omdpc B$ such that $\strat_k$ is no-lose from $x_k$, we can compute their union with $\strat$ to obtain a no-lose strategy in $\omdpc A \seqcomp \omdpc B$.

\begin{figure}[t]
    \begin{minipage}{.5\textwidth}
        \centering
        \begin{tikzpicture}[
    edge/.style={
        Arr,
        every node/.append style={ pos=0.5, },
    },
    s/.style={draw, circle, minimum size=0.4cm, inner sep=0},
    s1/.style={draw, circle, minimum size=0.2cm, inner sep=0},
    m/.style={insert path={[fill=black]circle[radius=1pt]}},
]

\newcommand{\mytransitions}[1]{
    \node[s] (en1) at (0.2,1.8){$\scalemath{0.7}{\enarg 1}$};
    \node[s] (en2) at (0.2,1.0){$\scalemath{0.7}{\enarg 2}$};
    \node[s] (en3) at (0.2,0.2){$\scalemath{0.7}{\enarg 3}$};
    \node[s] (ex1) at (1.8,1.8){$\scalemath{0.7}{\exarg 1}$};
    \node[s] (ex2) at (1.8,1.0){$\scalemath{0.7}{\exarg 2}$};
    \node[s] (ex3) at (1.8,0.2){$\scalemath{0.7}{\exarg 3}$};
}

\begin{scope}
    \draw[omdp] (0,0) rectangle ++(2,2);

    \mytransitions A
    
    \node (Alabel) at (1,2.25) {$\omdpc A$};

    \node[s1] (s1) at (1, 1.8) {};
    \node[s1] (s2) at (0.5, 0.5) {};
    \node[s1] (s3) at (1, 1) {};
            
    \draw[thick]
        (en2) -- ++(0.35,-0.15) [m]
            edge[edge] (s2)
        (s2) -- ++(0.75,0.25) [m]
            edge[edge, bend left=30] (s2)
            edge[edge] (ex3)
            edge[edge] (ex2);
            
    \draw[thick]
        (en1) -- ++(0.35,0) [m]
            edge[edge] (s1)
        (s1) -- ++(0.35,0) [m]
            edge[edge, bend right] (s1)
            edge[edge] (ex1)
        (s1) -- ++(0.35,-0.15) [m]
            edge[edge] (ex2)
            ;
            
    \draw[thick]
        (en3) -- ++(0.75,0) [m]
            edge[edge] (ex3);

    \draw[thick]
        (s3) -- ++(0.35, 0) [m]
            edge[edge, bend left] (s3);

    \draw[thick]
        (en2) -- ++(1, 0.35) [m]
            edge[edge] (s3)
            edge[edge] (ex1);

    \path 
        (en3) -- ++(0.75,0) [m]
        [draw] -- ++(0.25,0.15) node[above=-2pt]{$\buchiverts$};
\end{scope}

\begin{scope}[xshift=3cm]
    \draw[omdp] (0,0) rectangle ++(2,2);
    \node (Blabel) at (1,2.25) {$\solarg{\omdpc A}$};
    
    \mytransitions A
    
    \draw[thick] 
        (en2) -- (1,1) 
              edge[edge] (ex2)
              edge[edge] (ex3) node[inner sep=0, outer sep=0]{\contour{white}{$\strat_3$}}
        (en1) edge[edge] node{\contour{white}{$\strat_1$}} (ex1)
              edge[edge] node{\contour{white}{$\strat_2$}} (ex2)
        (en3) edge[edge] node{\contour{white}{$\strat_4^+$}} (ex3);
\end{scope}

\draw[double distance=1.5pt, -Implies] (2.25,1) -- node[above]{\tiny solution} (2.75,1);

\end{tikzpicture}
        \caption{Computing the local solution.}
        \label{fig:local_solution}
    \end{minipage}%
    \begin{minipage}{.5\textwidth}
        \centering
        \subfloat{
        \begin{tikzpicture}[]

\contourlength{0.5pt}

\begin{scope}
    \foreach \x in {0,...,5} {
        \draw[Arr] ($ (2,1.8) + \x*(0,-0.3) $) -- ++(0.5, 0);
    };
    
    \begin{scope}
        \clip (0,0) rectangle (2,2);
        \filldraw[dashed, fill=red!60, opacity=0.9, thick] (0, 1.8)  -- (2,1.9) -- ++(0, -0.5)[rounded corners=10pt] -- (1.5,1)[sharp corners] -- (2,0.7) -- ++(0,-0.5)[rounded corners=10pt] -- (0.8,0.8)[sharp corners] -- cycle;
    \end{scope}

    \draw[omdp] (0,0) rectangle ++(2,2);
    \node (Alabel) at (1,2.25) {$\omdpc A$};

    \draw[Arr, overlay] (-0.25, 1.8) node[left]{$\strat$} -- ++(0.25, 0);
\end{scope}

\begin{scope}[xshift=2.5cm]
    \foreach \x in {0,...,5} {
        \draw[Arr] ($ (2,1.8) + \x*(0,-0.3) $) -- ++(0.25, 0);
    };
    
    \foreach \x / \mycolor [count=\n] in {0/green, 1/blue, 4/orange, 5/purple} {
        \begin{scope}
            \clip (0,0) rectangle (2,2);
            \filldraw[dashed, fill=\mycolor!60, opacity=0.9, thick] ($ (0,1.8) - \x*(0,0.3) $) -- ++(2,0.1) -- ++(0,-0.3) -- cycle;
        \end{scope}
        \node at ($ (-0.75, 1.8) - \x*(0, 0.3) $) {\contour{white}{$x_\n$}};
        \node[overlay] at ($ ( 2.40, 1.8) - \x*(0, 0.3) $) {$Y_\n$};
        \node at ($ ( 0.50, 1.8) - \x*(0, 0.3) $) {\contour{white}{$\strat_\n$}};
    };
    
    \node (Blabel) at (1,2.25) {$\omdpc B$};
    \draw[omdp] (0,0) rectangle ++(2,2);
\end{scope}

\path (Alabel) -- node{$\seqcomp$} (Blabel);
\end{tikzpicture}
        }
        \caption{$\solarg{\omdpc A} \seqcomp \solarg{\omdpc B}$.}
        \label{fig:sequential_composition_solution}
    \end{minipage}
\end{figure}

We formalize the above intuition using the following \emph{decomposition lemma}.
\begin{restatable}[sequential decomposition]{lemma}{lemsequentialcomposition}\label{lem:sequential_composition}
Let $\strat$ be a no-lose strategy in $\omdpc A$ and let $\tuple{\{x_1, \dots, x_N\}, b} \defeq \solstrat^{\omdpc A}(i, \strat)$.
For $1 \le k \le N$, let $\strat_k$ be a no-lose strategy from $x_k$ in $\omdpc B$.
We have that
\begin{align*}
\solstrat^{\omdpc A \seqcomp \omdpc B}(i, \strat \cup \strat_1 \cup \dots \cup \strat_N) = \tuple{\emptyset, b} \sqcup \bigsqcup_{1 \le k \le N} \solstrat^{\omdpc B}(x_k, \strat_k).
\end{align*}
\end{restatable}\noindent
Based on \cref{lem:sequential_composition}, we define the sequential composition of solutions.

\begin{definition}[$\solarg 1 \seqcomp \solarg 2$]\label{def:semec_composition}
Let $\solarg 1$ and $\solarg 2$ be solutions such that $\typemdp{\solarg 1}\colon m \to l$ and $\typemdp{\solarg 2}\colon l \to n$.
We define $\solarg 1 \seqcomp \solarg 2$ for entrance $i$ as
\begin{align*}
    (\solarg 1 \seqcomp \solarg 2)(i) \defeq \left\{
        \tuple{\emptyset, b} \sqcup \bigsqcup_{1 \le k \le N} Y_k \midvert
            {\def\arraystretch{1.2}
            \begin{array}{cc}
                \tuple{\{ x_1, \dots, x_N\}, b} \in \solarg 1(i),\\
                Y_1 \in \solarg 2(x_1), \dots, Y_N \in \solarg 2(x_N)
            \end{array}}
    \right\}.
\end{align*}
\end{definition}
We conclude by outlining a proof of compositionality.
\begin{restatable}[compositionality of $\seqcomp$]{lemma}{lemcompofseq}\label{lem:comp_of_seq}
$\solarg{\omdpc A \seqcomp \omdpc B} = \solarg{\omdpc A} \seqcomp \solarg{\omdpc B}.$
\end{restatable}
\begin{proofs}
The $\supseteq$ direction follows from \cref{lem:sequential_composition}.
For the $\subseteq$ direction we show that we can always decompose a no-lose strategy in $\omdpc A \seqcomp \omdpc B$ into a no-lose strategy $\strat$ in $\omdpc A$ and no-lose strategies in $\omdpc B$ such that we can again apply \cref{lem:sequential_composition}.
\end{proofs}

\subsubsection{Sum composition.}
The sum of two solutions is simply their disjoint union. %
\begin{definition}[$\solarg1 \sumcomp \solarg2$]
Let $\solarg1$ and $\solarg2$ be solutions such that $\typemdp{\solarg1}\colon m_1 \to n_1$ and $\typemdp{\solarg2}\colon m_2 \to n_2$.
We define $(\solarg1 \sumcomp \solarg2)(i)$ as
$\solarg1(i)$ if $1 \le i \le m_1$ and $\solarg2(i-m_1)$ otherwise.
\end{definition}
The following statement follows straightforwardly.
\begin{restatable}[compositionality of $\sumcomp$]{lemma}{lemcompofsum}\label{lem:comp_of_sum}
$\solarg{\omdpc A \sumcomp \omdpc B} = \solarg{\omdpc A} \sumcomp \solarg{\omdpc B}.$
\end{restatable}%
\subsubsection{Trace.}
We explain the intuition of the trace of a solution using \cref{fig:trace_cases}. %

First, we pick a no-lose strategy $\strat_1$ in $\omdpc A$ and let $\tuple{T_1, b_1}$ be its effect.
Depending on the values of $T_1$ and $b_1$, we can distinguish three cases.
\begin{figure}[t]%
    \begin{minipage}{0.33\textwidth}%
    \subfloat[]{%
    \centering%
    \begin{tikzpicture}

\begin{scope}
\clip (0,0) rectangle ++(2,2);
\fill[fill=red!60, draw=black, dashed, thick, opacity=0.9] (0, 1.8) -- ++(2,0.1) -- ++(0,-0.8)[rounded corners=10pt] -- (0.8, 0.8)[sharp corners] -- cycle; %
\end{scope}

\draw[omdp] (0,0) rectangle ++(2,2);

\draw[Arr] (-0.25, 1.8) node[left]{$\strat_1$} -- ++(0.25, 0);
\draw[Arr] (2, 1.8) -- ++(0.25, 0);
\draw[Arr] (2, 1.5) -- ++(0.25, 0);
\draw[Arr] (2, 1.2) -- ++(0.25, 0);
\draw[Arr] (2, 0.9) -- ++(0.25, 0);
\draw[Arr] (2, 0.6) -- ++(0.25, 0);
\draw[Arr, rounded corners] (2, 0.25) -| ++(0.25, -0.5) -| ++(-2.5, 0.5) -- ++(0.25, 0);

\draw[thick, decoration={brace}, decorate] (2.3, 1.8) -- (2.3, 1.2) node[pos=0.5,right]{$T_1$};

\end{tikzpicture}%
    \label{fig:trace1}
    }%
    \end{minipage}%
    \begin{minipage}{0.33\textwidth}%
    \subfloat[]{%
    \centering%
    \begin{tikzpicture}

\begin{scope}
\clip (0,0) rectangle ++(2,2);
\fill[fill=orange!60, draw=black, dashed, thick] (0, 0.25)[rounded corners=10pt] -- (1, 0.5)[sharp corners] -- (2, 0.4) -- ++(0,-0.3) -- cycle;
\fill[fill=red!60, draw=black, dashed, thick, opacity=0.9] (0, 1.8) -- ++(2,-1.05) -- ++(0,-0.65) -- cycle;
\end{scope}

\draw[omdp] (0,0) rectangle ++(2,2);

\draw[Arr] (-0.25, 1.8) node[left]{$\strat_1$} -- ++(0.25, 0);
\draw[Arr] (2, 1.8) -- ++(0.25, 0);
\draw[Arr] (2, 1.5) -- ++(0.25, 0);
\draw[Arr] (2, 1.2) -- ++(0.25, 0);
\draw[Arr] (2, 0.9) -- ++(0.25, 0);
\draw[Arr] (2, 0.6) -- ++(0.25, 0);
\draw[Arr, rounded corners] (2, 0.25) -| (2.25, -0.25) -| (-0.25, 0.25) -- (0, 0.25) node[left, xshift=-4pt]{$\strat_2$};

\filldraw (0.5, 0.3) circle (1pt) node[above] {$\buchiverts$};

\draw[overlay, thick, decoration={brace}, decorate] (2.3, 0.25) -- ++(0, -0.6) node[pos=0.5,right]{$\substack{T_1\\T_2}$};
\draw[thick, decoration={brace}, decorate] (2.3, 0.6) -- ++(0, -0.3) node[pos=0.5,right]{$T_1$};
\node[overlay] at (1,-0.2) {\contour{white}{$\strat'$}};

\end{tikzpicture}%
    \label{fig:trace3}
    }%
    \end{minipage}%
    \begin{minipage}{0.33\textwidth}%
    \subfloat[]{%
    \centering%
    \begin{tikzpicture}

\begin{scope}
\clip (0,0) rectangle ++(2,2);
\fill[fill=orange!60, draw=black, dashed, thick] (0, 0.25) -- ++(2,1.4) -- ++(0,-0.6)[rounded corners=10pt] -- (1, 0.5)[sharp corners] -- (2, 0.4) -- ++(0,-0.3) -- cycle;
\fill[fill=red!60, draw=black, dashed, thick, opacity=0.9] (0, 1.8) -- ++(2,-1.05) -- ++(0,-0.65) -- cycle;
\end{scope}

\draw[omdp] (0,0) rectangle ++(2,2);

\draw[Arr] (-0.25, 1.8) node[left]{$\strat_1$} -- ++(0.25, 0);
\draw[Arr] (2, 1.8) -- ++(0.25, 0);
\draw[Arr] (2, 1.5) -- ++(0.25, 0);
\draw[Arr] (2, 1.2) -- ++(0.25, 0);
\draw[Arr] (2, 0.9) -- ++(0.25, 0);
\draw[Arr] (2, 0.6) -- ++(0.25, 0);
\draw[Arr, rounded corners] (2, 0.25) -| (2.25, -0.25) -| (-0.25, 0.25) -- (0, 0.25) node[left, xshift=-4pt]{$\strat_2$};

\draw[overlay, thick, decoration={brace}, decorate] (2.3, 0.25) -- ++(0, -0.6) node[pos=0.5,right]{$\substack{T_1\\T_2}$};
\draw[thick, decoration={brace}, decorate] (2.3, 1.5) -- ++(0, -0.3) node[pos=0.5,right]{$T_2$};
\draw[thick, decoration={brace}, decorate] (2.3, 0.6) -- ++(0, -0.3) node[pos=0.5,right]{$T_1$};
\node[overlay] at (1,-0.2) {\contour{white}{$\strat'$}};

\end{tikzpicture}%
    \label{fig:trace2}
    }%
    \end{minipage}%
    \caption{Three ways of obtaining an effect in $\tr(\sol)$.}%
    \label{fig:trace_cases}%
\end{figure}
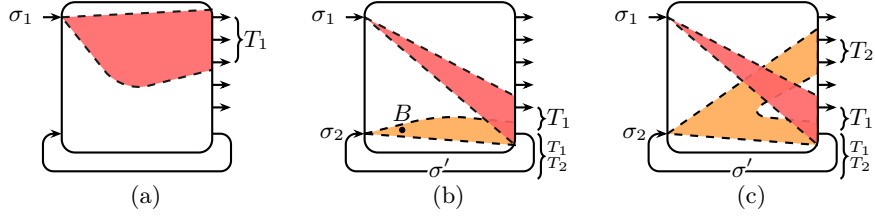%
\begin{enumerate}%
    \item[\subref{fig:trace1}:] If $\exarg{n+1}$ (the \emph{trace exit}) is not in $T_1$, then $\strat_1$ is also a no-lose strategy in $\tr(\omdpc A)$ with the same effect.
\end{enumerate}
    Otherwise, for the remaining two cases, we need a no-lose strategy $\strat_2$ from entrance $m+1$.
    Let $\tuple{T_2, b_2}$ be the effect of $\strat_2$.
\begin{enumerate}%
    \item[\subref{fig:trace3}:] If $T_2 = \Set*{ \exarg{n+1} }$, then $b_2$ must be true.
    Otherwise it would imply that $\strat_2$ forms a \emph{losing cycle} in $\tr(\omdpc A)$.
    \item[\subref{fig:trace2}:] If $T_2 \ne \Set*{ \exarg{n+1} }$, then the union of $\strat_1$ and $\strat_2$ forms a no-lose strategy in $\tr(\omdpc A)$.
    If additionally we have that $\exarg{n+1} \in T_2$, then, intuitively, the formed cycle ensures that we eventually reach an exit.
\end{enumerate}
We formalize the above intuition using the following \emph{decomposition lemma}.
\begin{restatable}[trace decomposition]{lemma}{lemtracedecomposition}\label{lem:trace_decomposition}
Let
\begin{itemize}
    \item $\omdpc A$ be an roMDP such that $\typemdp{\omdpc A}\colon m+1 \to n+1$, and $i \in [m]$,
    \item $\strat_1$, $\strat_2$ be no-lose strategies in $\omdpc A$ for entrances $i$ and $m+1$, respectively,
    \item $\tuple{T_1, b_1} \defeq \solstrat^{\omdpc A}\pars*{i, \strat_1}$ such that $\exarg{n+1} \in T_1$,
    \item $\tuple{T_2, b_2} \defeq \solstrat^{\omdpc A}\pars*{m+1, \strat_2}$ such that $\tuple{T_2, b_2} \ne \tuple{\{\exarg{n+1}\}, \bot}$,
    \item $\strat' \defeq \Set*{ \exarg{n+1} \mapsto \{*\} } \cup \Set*{ v \mapsto \emptyset \mid v \in \verts^{\tr(\omdpc A)}, v \ne \exarg{n+1} }$.
\end{itemize}
Then,
$\solstrat^{\tr(\omdpc A)}\pars*{i, \strat_1 \cup \strat_2 \cup \strat'} = \tuple*{(T_1 \cup T_2) \setminus \Set*{ \exarg{n+1} }, b_1 \vee b_2}.$
\end{restatable}\noindent
Based on \cref{lem:trace_decomposition} we define the trace of a solution.

\begin{definition}[$\tr(\sol)$]
Let $\sol$ be a solution such that $\typemdp{\sol}\colon m+1 \to n+1$.
We define $\tr(\sol)$ with $\typemdp{\tr(\sol)}\colon m \to n$ as %
\begin{align*}
    &\big(\tr(\sol)\big)(i) \defeq \big\{ \tuple{T, b} \in \sol(i) \mid n+1 \not\in T \big\} \cup {}\\
    &\bigcup_{\substack{\tuple{T_1, b_1} \in \sol(i)\\\text{s.t.}\\n+1 \in T_1}}
    \left\{ 
        \tuple*{(T_1 \cup T_2) \setminus \{n+1\}, b_1 \vee b_2}
        \,\middle|\,
        \begin{array}{ll}
        \tuple{T_2, b_2} \in \sol(m+1) \text{ s.t. }\\
        \tuple{T_2, b_2} \ne \tuple{\{n+1\}, \bot}
        \end{array}
    \right\}.
\end{align*} 
\end{definition}
We conclude by outlining a proof of compositionality.
\begin{restatable}[compositionality of $\tr$]{lemma}{lemcompoftr}\label{lem:comp_of_tr}
$\solarg{\tr(\omdpc A)} = \tr(\solarg{\omdpc A}).$
\end{restatable}
\begin{proofs}
The $\supseteq$ direction of the proof follows from \cref{lem:trace_decomposition}.
For the $\subseteq$ direction of the proof, we take a no-lose strategy in $\tr(\omdpc A)$ and show that we can decompose it into $\strat_1$, $\strat_2$, $\strat'$, and again apply \cref{lem:trace_decomposition}.
We show that taking the loop of the traced roMDP multiple times does not change the reachable states.
\end{proofs}

\subsection{One-shot Bottom-up Algorithm}\label{sec:bottomupalg}
\Cref{thm:semsol_compositional} gives rise to a natural, bottom-up computation.
\ifwithappendix
We provide pseudocode for the bottom-up algorithm in Appendix~\ref{sec:one_shot_bottom_up_algorithm}.
\else
We provide pseudocode for the bottom-up algorithm in Appendix~B.3 of the technical report~\cite{technicalreport}.
\fi
Its correctness follows directly from \cref{thm:semsol_compositional,lem:local_solution_correct}.
Because our solution is compositional, we can make use of \emph{solution sharing}.

\begin{corollary}[solution sharing]\label{cor:solution_sharing}
Let $\omdpc A$, $\omdpc B$ and $\omdpc C$ be roMDPs.
If $\solarg{\omdpc A} = \solarg{\omdpc B}$ then
\begin{align*}
\begin{array}{cc}
\solarg{\omdpc A \seqcomp \omdpc C} = \solarg{\omdpc B \seqcomp \omdpc C},\\
\solarg{\omdpc C \seqcomp \omdpc A} = \solarg{\omdpc C \seqcomp \omdpc B},
\end{array}\qquad
\begin{array}{cc}
\solarg{\omdpc A \sumcomp \omdpc C} = \solarg{\omdpc B \sumcomp \omdpc C},\\
\solarg{\omdpc C \sumcomp \omdpc A} = \solarg{\omdpc C \sumcomp \omdpc B},
\end{array}\qquad
\solarg{\tr(\omdpc A)} = \solarg{\tr(\omdpc B)}.
\end{align*}
\end{corollary}
Thus, an efficient implementation of the bottom-up algorithm only needs to compute the solution of each unique subtree of $\sd$ (up to solution equality).
If we limit the number of exits each roMDP can have to a constant, the maximum number of effects that a solution can have is also a constant, yielding a \emph{polynomial-time} algorithm.
However, in general, the running time of the algorithm grows exponentially in the number of exits.

\section{Strategy Refinement Algorithm}\label{sec:refinement}
The bottom-up algorithm above computes the set of all effects, which can grow exponentially in the number of exits.
Below, we present an algorithm for solving almost-sure Büchi by iteratively reasoning about at most one effect per entrance.

First, we introduce the \emph{shortcut graph}, adapting the shortcut MDP from~\cite{DBLP:conf/tacas/WatanabeVHRJ24}.
This graph contains all entrances of $\sd$, connected by transitions which represent the effect of some local strategy (the \emph{shortcuts}).
We show that an entrance in this shortcut graph is winning iff it is winning in $\sem\sd$.

Next, we show that for each entrance the set of all effects forms a \emph{join semilattice}.
This allows us to define the \emph{maximum effect} given some subset of `allowed exits' (akin to the role of $Y$ in the classical Büchi operator).
We define a specialized \emph{refinement Büchi operator} that uses maximum effects to prevent constructing the entire shortcut graph, yielding a polynomial time algorithm.

\subsection{Shortcut Graph}\label{sec:shortcut_graph}
We need some additional definitions to reason about string diagrams.
\begin{definition}%
The set of \emph{component entrances} is inductively defined by the following:
$\componentsent(\omdpc A) \defeq \Set*{ \en^{\omdpc A} }$,
$\componentsent(\omdpc A * \omdpc B) \defeq \componentsent(\omdpc A) \uplus \componentsent(\omdpc B)$ for
$* \in \{\seqcomp, \sumcomp\}$, and
$\componentsent(\tr(\omdpc A)) \defeq \componentsent(\omdpc A)$.
\end{definition}
We define a \emph{connection mapping} $\mappingfunc{\sd}{\omdpc A}\colon \powerset{\ex^{\omdpc A}} \to \powerset{\componentsent(\sd)}$ for each roMDP $\omdpc A$, which maps a set of (local) exits $X \subseteq \ex^{\omdpc A}$ to the connected component entrances in $\sd$ (exits that are not connected to any entrance do not map to any entrance). %
A precise definition of $\mappingfunc{\sd}{\omdpc A}$ straightforwardly follows by induction on the construction of $\sem\sd$.\sj{Is this in the appendix? At least for your thesis, I would recommend having this written out :) Marck: Currently no, but I agree that I want it in my thesis.}
\begin{definition}[shortcut graph]\label{def:shortcut_graph}
The \emph{shortcut graph} $\shortcut$ is a graph $\tuple{\verts_1, \verts_P, \edges}$ where
$\verts_1 \defeq \componentsent(\sd)$,
$\verts_P \defeq \bigcup_{i \in \componentsent(\sd)} \Set*i \times \sol_{\omdpc A}(i)$ and $\edges$ satisfies the following equations for all $\tuple*{\enarg i^{\omdpc A}, \tuple{X, b}} \in \verts_P$: 
\begin{gather*}
\Pre\pars*{\tuple*{\enarg i^{\omdpc A}, \tuple{X, b}}} = \Set*{\enarg i^{\omdpc A}}, \quad
\Post\pars*{\tuple*{\enarg i^{\omdpc A}, \tuple{X, b}}} = \begin{cases}
    \mappingfunc{\sd}{\omdpc A}(X) & \text{if } X \ne \emptyset\\
    \Set*{\enarg i^{\omdpc A}} & \text{if } X = \emptyset
\end{cases}%
\end{gather*}
The Büchi vertices (of $\shortcut$) are $\buchiverts_\sd \defeq \{ \tuple{i, \tuple{X, b}} \in \verts_P \mid b=\top \}$.
\end{definition}
Intuitively, the shortcut graph transforms each leaf roMDP $\omdpc A$ of the string diagram into a graph representing the effects present in $\solarg{\omdpc A}$.
These `leaf graphs' are then combined using the usual operational semantics of the string diagram.
Note that the size of the shortcut graph, like the number of effects in the solution of an roMDP, is exponential in the number of exits.

\begin{figure}[t]
    \begin{minipage}{.32\textwidth}%
        \centering%
        \subfloat[][]{%
            \resizebox{\textwidth}{!}{\begin{tikzpicture}[
    edge/.style={
        Arr,
        every node/.append style={ pos=0.5, },
    },
    s/.style={draw, circle, minimum size=0.4cm, inner sep=0},
]

\contourlength{0.5pt}

\newcommand{\mytransitions}[1]{
    \node[s] (en1) at (0.2,1.8){$\scalemath{0.7}{\enarg 1}$};
    \node[s] (en2) at (0.2,1.0){$\scalemath{0.7}{\enarg 2}$};
    \node[s] (en3) at (0.2,0.2){$\scalemath{0.7}{\enarg 3}$};
    \node[s] (ex1) at (1.8,1.8){$\scalemath{0.7}{\exarg 1}$};
    \node[s] (ex2) at (1.8,1.0){$\scalemath{0.7}{\exarg 2}$};
    \node[s] (ex3) at (1.8,0.2){$\scalemath{0.7}{\exarg 3}$};
}

\begin{scope}
    \draw[omdp] (0,0) rectangle ++(2,2);
    \mytransitions A
    
    \node at (1,2.25) {$\omdpc A$};
    \draw[Arr] (-0.25, 1.8) -- ++(0.25, 0);
    \draw[Arr] (2, 1.8) -- ++(0.5, 0);
    \draw[Arr] (2, 1) -- ++(0.5, 0);
    \draw[Arr] (2, 0.2) -- ++(0.5, 0);

    \draw[thick] 
        (en2) -- (1,1) 
              edge[edge] (ex2)
              edge[edge] (ex3) node[inner sep=0, outer sep=0]{\contour{white}{$\strat_3$}}
        (en1) edge[edge] node{\contour{white}{$\strat_1$}} (ex1)
              edge[edge] node{\contour{white}{$\strat_2$}} (ex2)
        (en3) edge[edge] node{\contour{white}{$\strat_4^+$}} (ex3);
        
\end{scope}

\begin{scope}[xshift=2.5cm]
    \draw[omdp] (0,0) rectangle ++(2,2);
    \mytransitions B
    
    \node at (1,2.25) {$\omdpc B$};
    \draw[Arr] (2, 1.8) -- ++(0.25, 0);
    \draw[Arr, rounded corners] (2, 0.2) -| ++(0.25, -0.5) -| ++(-5,0.5) -- ++(0.25,0);
    \draw[Arr, rounded corners] (2, 1) -| ++(0.5, -1.5) -| ++(-5.5, 1.5) -- ++(0.5,0);
    
    \draw[edge]
        (en1) edge node{\contour{white}{$\strat_5$}} (ex1)
        (en2) edge node{\contour{white}{$\strat_6$}} (ex1)
              edge node{\contour{white}{$\strat_7$}} (ex2)
              edge node{\contour{white}{$\strat_8$}} (ex3)
        (en3) edge node{\contour{white}{$\strat_9$}} (ex3);
\end{scope}
\end{tikzpicture}}%
            \label{fig:shortcut1}
        }%
    \end{minipage}\hfill%
    \begin{minipage}{.32\textwidth}%
        \centering%
        \subfloat[][]{%
            \resizebox{\textwidth}{!}{\begin{tikzpicture}[
    edge/.style={
        Arr,
        every node/.append style={ pos=0.5, },
    },
    s/.style={draw, circle, minimum size=0.4cm, inner sep=0},
]

\contourlength{0.5pt}

\newcommand{\mytransitions}[1]{
    \node[s] (en1) at (0.2,1.8){$\scalemath{0.7}{\enarg 1}$};
    \node[s] (en2) at (0.2,1.0){$\scalemath{0.7}{\enarg 2}$};
    \node[s] (en3) at (0.2,0.2){$\scalemath{0.7}{\enarg 3}$};
    \node[s] (ex1) at (1.8,1.8){$\scalemath{0.7}{\exarg 1}$};
    \node[s] (ex2) at (1.8,1.0){$\scalemath{0.7}{\exarg 2}$};
    \node[s] (ex3) at (1.8,0.2){$\scalemath{0.7}{\exarg 3}$};
}

\begin{scope}
    \draw[omdp] (0,0) rectangle ++(2,2);
    \mytransitions A
    
    \node at (1,2.25) {$\omdpc A$};
    \draw[Arr] (-0.25, 1.8) -- ++(0.25, 0);
    \draw[Arr] (2, 1.8) -- ++(0.5, 0);
    \draw[Arr] (2, 1) -- ++(0.5, 0);
    \draw[Arr] (2, 0.2) -- ++(0.5, 0);

    \draw [thick]
        (en2) -- (1,1)
              edge[edge] (ex2)
              edge[edge] (ex3)
        (en2) edge[edge] node{\contour{white}{$\strat_3$}} (ex2)
        (en3) edge[edge] node{\contour{white}{$\strat_4^+$}} (ex3)
        (en1) -- (1,1.5)
            edge[edge] (ex1)
            edge[edge] (ex2) node[inner sep=0, outer sep=0, minimum size=0, pos=1]{\contour{white}{$\strat_{1\cup2}$}};
\end{scope}

\begin{scope}[xshift=2.5cm]
    \draw[omdp] (0,0) rectangle ++(2,2);
    \mytransitions B
    
    \node at (1,2.25) {$\omdpc B$};
    \draw[Arr] (2, 1.8) -- ++(0.25, 0);
    \draw[Arr, rounded corners] (2, 0.2) -| ++(0.25, -0.5) -| ++(-5,0.5) -- ++(0.25,0);
    \draw[Arr, rounded corners] (2, 1) -| ++(0.5, -1.5) -| ++(-5.5, 1.5) -- ++(0.5,0);
    
    \draw[thick]
        (en1) edge[edge] node{\contour{white}{$\strat_5$}} (ex1)
        (en3) edge[edge] node{\contour{white}{$\strat_9$}} (ex3)
        (en2) -- (1,1)
            edge[edge] (ex1)
            edge[edge] (ex2)
            edge[edge] (ex3) node{\contour{white}{$\strat_{6\cup7\cup8}$}};
\end{scope}
\end{tikzpicture}}%
            \label{fig:shortcut2}
        }%
    \end{minipage}\hfill%
    \begin{minipage}{.32\textwidth}
        \centering%
        \subfloat[][]{%
            \resizebox{\textwidth}{!}{\begin{tikzpicture}[
    edge/.style={
        Arr,
        every node/.append style={ pos=0.5, },
    },
    s/.style={draw, circle, minimum size=0.4cm, inner sep=0},
    l/.style={fill=red!20},
]

\contourlength{0.5pt}

\begin{scope}
    \draw[omdp] (0,0) rectangle ++(2,2);
    \node[s] (en1) at (0.2,1.8){$\enarg 1$};
    \node[s] (en2) at (0.2,1.0){$\enarg 2$};
    \node[s] (en3) at (0.2,0.2){$\enarg 3$};
    \node[s,l] (ex1) at (1.8,1.8){$\exarg 1$};
    \node[s] (ex2) at (1.8,1.0){$\exarg 2$};
    \node[s] (ex3) at (1.8,0.2){$\exarg 3$};
    
    \node at (1,2.25) {$\omdpc A$};
    \draw[Arr] (-0.25, 1.8) -- ++(0.25, 0);
    \draw[Arr] (2, 1.8) -- ++(0.5, 0);
    \draw[Arr] (2, 1) -- ++(0.5, 0);
    \draw[Arr] (2, 0.2) -- ++(0.5, 0);

    \draw
        (en2) -- (1,1)
              edge[edge] (ex2)
              edge[edge] (ex3)
        (en1) edge[edge] node{\contour{white}{$\strat_2$}} (ex2)
        (en2) edge[edge] node{\contour{white}{$\strat_3$}} (ex2)
        (en3) edge[edge] node{\contour{white}{$\strat_4^+$}} (ex3);
        
\end{scope}

\begin{scope}[xshift=2.5cm]
    \draw[omdp] (0,0) rectangle ++(2,2);
    \node[s,l] (en1) at (0.2,1.8){$\enarg 1$};
    \node[s] (en2) at (0.2,1.0){$\enarg 2$};
    \node[s] (en3) at (0.2,0.2){$\enarg 3$};
    \node[s,l] (ex1) at (1.8,1.8){$\exarg 1$};
    \node[s] (ex2) at (1.8,1.0){$\exarg 2$};
    \node[s] (ex3) at (1.8,0.2){$\exarg 3$};
    
    \node at (1,2.25) {$\omdpc B$};
    \draw[Arr] (2, 1.8) -- ++(0.25, 0);
    \draw[Arr, rounded corners] (2, 0.2) -| ++(0.25, -0.5) -| ++(-5,0.5) -- ++(0.25,0);
    \draw[Arr, rounded corners] (2, 1) -| ++(0.5, -1.5) -| ++(-5.5, 1.5) -- ++(0.5,0);
    
    \draw[thick]
        (en2) -- (1,1) 
            edge[edge] node[pos=0]{\contour{white}{$\strat_{7\cup8}$}} (ex2)
            edge[edge] (ex3)
        (en3) edge[edge] node{\contour{white}{$\strat_9$}} (ex3);
\end{scope}
\end{tikzpicture}}%
            \label{fig:shortcut3}
        }%
        \end{minipage}%
    \caption{Example refinement algorithm execution.}
    \label{fig:shortcut}
\end{figure}
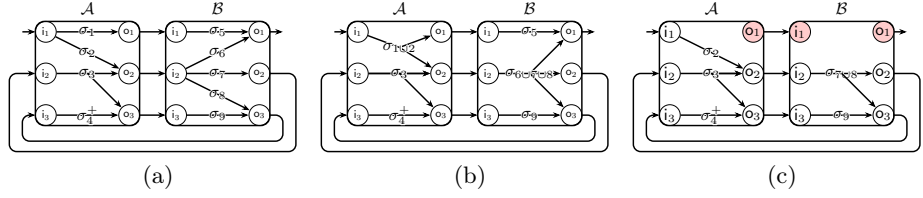
\begin{myexample}\label{ex:shortcut}
    \cref{fig:shortcut1} depicts the shortcut graph of $\omdpc C$ in \cref{fig:example_omdp}.
    Each arrow represents a no-lose strategy $\strat$ that can be played.
    The vertices reached by the arrow depict the proper exit set of $\strat$.
    In the figure, we have omitted no-lose strategies that can be obtained by taking the union of two no-lose strategies.
    Beside a strategy we write a $+$ to indicate that a Büchi state is reachable.
    We see that $\solstrat^{\omdpc A}(\enarg2, \strat_3) = \tuple{\{\exarg2, \exarg3\}, \bot}$ and $\solstrat^{\omdpc A}(\enarg3, \strat_4) = \tuple{\{\exarg3\}, \top}$.
    To satisfy the Büchi condition, we must take the transition from $\enarg3^{\omdpc A}$ to $\exarg3^{\omdpc A}$ infinitely often.
\end{myexample}
The shortcut graph preserves the winning region of the entrances:
\begin{restatable}{theorem}{lemshortcutgraph}\label{lem:shortcut_graph}
For all $i \in \componentsent(\sd)$ we have $i \models^{\sem\sd} \buchiobj\buchiverts \text{ iff } i \models^\shortcut \buchiobj\buchiverts_\sd.$
\end{restatable}\noindent%
\ifwithappendix
The proof is given in \cref{app:shortcut}.
\else
The proof is given in the appendix of the technical report~\cite{technicalreport}.
\fi

\subsection{Strategy Refinement Algorithm}\label{sec:refinement_algorithm}
In our novel strategy refinement algorithm, we exploit the fact that the effects of a given entrance are a \emph{join semilattice}.
In short, the reason is that no-lose strategies are closed under union and $\solarg{\omdpc A}(i)$ is finite. %
\begin{restatable}[no-lose strategy union]{lemma}{lemstrategyunion}\label{lem:strategy_union}
Let $\strat_1, \strat_2$ be no-lose strategies from entrance $i$.
We have that $\solstrat(i, \strat_1 \cup \strat_2) = \solstrat(i, \strat_1) \sqcup \solstrat(i, \strat_2)$.
\end{restatable}
\begin{lemma}%
\label{lem:lattice}
$\tuple{\solarg{\omdpc A}(i), \sqsubseteq}$ is a \emph{join semilattice} where $\sqsubseteq$ is the lexicographical order, that is, $\tuple{T_1, b_1} \sqsubseteq \tuple{T_2, b_2} \text{ iff } T_1 \subset T_2 \vee [T_1 \subseteq T_2 \wedge b_1 \le b_2]$.
\end{lemma}
Using the join semilattice structure of effects, we now define \emph{maximum effects}.
\begin{definition}[maximum effect]\label{def:maximum_effect}
The \emph{maximum effect for $i$ restricted to}\sj{Can we give the notation a more descriptive thing, maybe use max here?, or $|E$? Marck: TODO, maybe for thesis?} $E \subseteq [n]$ is $\solarg{\omdpc A}(i, E) \defeq \max \{ \tuple{E', b} \in \solarg{\omdpc A}(i) \mid E' \subseteq E \}$ where $\max\emptyset \defeq \tuple{\emptyset, \bot}$.%
\end{definition}
Note that the maximum effect restricted to $E$ is equal to the join of all effects that have a proper exit set that is a subset of $E$, and therefore always exists.
Using this notion of maximum effect, we define the \emph{refinement Büchi operator} which is specialized for shortcut graphs.

\begin{definition}[refinement Büchi operator]
Let $X, Y \subseteq \verts_1^{\graph_\sd}$, we define
\begin{align*}
F'(X, Y) \defeq \left\{\enarg i^{\omdpc A} \in \componentsent(\sd) \middle|
\begin{array}{ll}
\text{let}  &\tuple{T, b} \defeq \solarg{\omdpc A}(i, Y)\\
\text{s.t.} &\tuple{T, b} \ne \tuple{\emptyset, \bot} \text{ and } b \vee \mappingfunc{\sd}{\omdpc A}(T) \cap X \ne \emptyset
\end{array}
\right\}.
\end{align*}
\end{definition}%
This new operator closely follows the classic Büchi operator defined in \cref{sec:prelims}.
We will now show that the new refinement Büchi operator matches the classic Büchi operator when applied to the shortcut graph.

\begin{restatable}[refinement]{theorem}{thmrefinement}\label{thm:refinement}
On $\shortcut$, $F'(X, Y) = F(X, Y)$ for all $X, Y \subseteq V_1$. %
\end{restatable}
\begin{myproof}
Starting with the definition of $F$ (\cref{def:buchi_operator}), we expand its definition in a shortcut graph and work towards the definition of $F'$.
\begin{align*}
&F(X, Y) \defeq \Set*{ v \in \verts_1 \midvert \exists v' \in \Post(v)\colon \Post(v') \subseteq Y \text{ and } v' \in \buchiverts \vee v' \in \Pre(X) }\\
\eqcomment{expand shortcut graph definition}
&{}= \Bigl\{ \enarg i^{\omdpc A} \in \verts_1 \mid \exists \tuple{T, b} \in \solarg{\omdpc A}(i)\colon{}
\mappingfunc{\sd}{\omdpc A}(T) \subseteq Y \text{ and } b \vee \mappingfunc{\sd}{\omdpc A}(T) \cap X \ne \emptyset \Bigr\}\\
\eqcomment{an effect satisfies the previous condition iff the maximum effect does}
&{}=
\left\{
\enarg i^{\omdpc A} \in \verts_1 \middle|
\begin{array}{ll}
\text{let} &\tuple{T,b} \defeq \max\Set*{ \tuple{T', b'} \in \solarg{\omdpc A}(i) \mid T \subseteq Y} \\
\text{s.t.} &\tuple{T,b} \ne \tuple{\emptyset, \bot} \text{ and } b \vee \mappingfunc{\sd}{\omdpc A}(T) \cap X \ne \emptyset
\end{array}
\right\}
\\
\eqcomment{use \cref{def:maximum_effect} of maximum effect}
&{}=
\left\{\enarg i^{\omdpc A} \in \verts_1 \middle|
\begin{array}{ll}
\text{let} &\tuple{T, b} \defeq \solarg{\omdpc A}(i, Y)\\
\text{s.t.} &\tuple{T,b} \ne \tuple{\emptyset, \bot} \text{ and } b \vee \mappingfunc{\sd}{\omdpc A}(T) \cap X \ne \emptyset
\end{array}
\right\}
= F'(X, Y)
\end{align*}
\end{myproof}
As a consequence of \cref{thm:refinement}, $F$ can be replaced by $F'$ in the fixpoint computation of \cref{lem:buchi}.
Our new strategy refinement algorithm is then simply the Kleene iteration of $F'$.
Crucially, this algorithm does \emph{not} require an explicit representation of the shortcut graph.

\begin{myexample}\label{ex:refinement_alg}
    We again consider \cref{fig:shortcut}
    and follow the Kleene iteration steps as described below \cref{lem:buchi}, but using the refinement Büchi operator $F'$ instead.
    Initially, $Y$ is equal to $\verts_1$. 
    The maximum effects considered by $F'$ for each entrance are depicted in~\subref{fig:shortcut2}.
    The next value of $Y$ (after computing the `inner fixpoint', i.e., $\lfpp X.F'(X,Y)$) is equal to all vertices that can reach a Büchi vertex by playing the depicted effects.
    Entrances that cannot reach a Büchi vertex by playing the maximum effects must be losing; in this example, this is only $\enarg 1^{\omdpc B}$.
    Therefore, $\enarg1^{\omdpc B}$ is not in the updated value of $Y$.
    In the next iteration of computing the inner fixpoint, we use the updated value for $Y$, see \subref{fig:shortcut3}.
    For $\enarg1^{\omdpc A}$, we compute the maximum effect not reaching $\exarg1$, which is $\tuple*{\Set*{\exarg 2}, \bot}$.
    Similarly, for $\enarg2^{\omdpc B}$, we compute the maximum effect not reaching $\exarg1$, which is $\tuple*{\Set*{\exarg2, \exarg3}, \bot}$.
    Finally, all entrances in $Y$ can reach a Büchi vertex and $Y$ is the greatest fixpoint of $F'$.
    Consequently, $Y$ is equal to the winning region of the shortcut graph.
\end{myexample}
In the classical Büchi algorithm, the number of times $F$ is computed (and therefore $F'$ in the refinement algorithm) is polynomial in the size of the MDP~{\cite{DBLP:conf/csl/ChatterjeeJH03}}.
Thus, the strategy refinement algorithm runs in polynomial time if we can compute maximum effects in polynomial time, which we show below. %
We also outline how intermediate results of the refinement algorithm can be cached.

\subsubsection{Computation of Maximum Effect.}\label{sec:biggest_effect_computation}
We briefly sketch how to efficiently compute the maximum effect for entrance $i$ restricted to some exit set $E \subseteq \ex^{\omdpc A}$ in roMDP $\omdpc A$.
First, we compute the winning region $W$ of the almost-sure Büchi objective $\buchiobj{(E \cup \buchiverts)}$, which can be done in polynomial time~\cite{DBLP:conf/csl/ChatterjeeJH03}.
These are exactly the vertices that have a no-lose strategy that does not reach exits $\ex^{\omdpc A} \setminus E$.
If $W$ includes $i$, then the strategy $\strat_W$ that randomizes over vertices that remain in $W$ reaches the biggest set of vertices and is no-lose.
Hence, $\solarg{\omdpc A}(i, E) = \solstrat\pars*{i, \strat_W}$.

\subsubsection{Caching.}\label{sec:refinement_caching}
The main computation of the refinement algorithm is that of maximum effects, which we try to cache. 
As an initial cache, we simply store the maximum effects and the inputs that were used to obtain it.
For example, if at some point during the refinement algorithm we compute that $\solarg{\omdpc A_1}(i, E) = E'$, then for some other occurrence of $\omdpc A$ in the string diagram, say, $\omdpc A_2$, we use the cache to obtain that $\solarg{\omdpc A_2}(i, E) = E'$.
Here we assume that $\buchiverts$ is defined equally for repeated occurrences of a component (otherwise, the cache may not be applied as is).

We can improve our caching by observing that the maximum effect is \emph{monotone} for a fixed entrance $i$:
\begin{restatable}{lemma}{lembiggestsolutionmonotone}%
\label{lem:biggest_solution_monotone}
If $E \subseteq E'$ then $\solarg{\omdpc A}(i, E) \sqsubseteq \solarg{\omdpc A}(i, E')$.
\end{restatable}\noindent%
We can use this monotonicity to cache the maximum effect of some exit set that is `in between' two previous computations, formalized by the following lemma:%
\begin{restatable}{lemma}{lemsolutionreuse}\label{lem:solution_reuse}
Let $\tuple{Y', b} \defeq \solarg{\omdpc A}(i, Y)$.
If $Y' \subseteq Y'' \subseteq Y$ then $\solarg{\omdpc A}(i, Y'') = \tuple{Y', b}$.
\end{restatable}

\section{Related Work}
Compositional model checking has been an active area of research for several decades, originating with the seminal work of Clarke et al.~\cite{ClarkeLM89}.
In this context, we focus on two closely related lines of inquiry: compositional probabilistic model checking and compositional algorithms for two-player infinite games.
Other forms of compositions, such as parallel compositions, have also been studied~\cite{FengKP10,FengKP11}.

\paragraph{Sequential Compositional Probabilistic Model Checking.}
Compositional probabilistic model checking using string diagrams of MDPs has been developed in a series of works~\cite{DBLP:conf/cav/WatanabeVJH24,DBLP:conf/cav/WatanabeEAH23,DBLP:conf/tacas/WatanabeVHRJ24}. 
These studies focus on quantitative properties such as reachability probabilities and rewards.
Likewise, model checking of hierarchical MDPs~\cite{HauskrechtMKDB98} is also closely related to our approach.
In particular, recent works~\cite{JungesS22,NearyVCT22} incorporate parameter synthesis techniques into hierarchical model checking algorithms.
None of the works above considers repeated reachability. 
Sequentially composed MDPs have also been studied in the context of learning-based approaches, such as in~\cite{JothimuruganBBA21,DBLP:conf/ifaamas/DelgrangeAL0N025}.

\paragraph{Compositional Algorithms for Two-Player Infinite Games.}
The works~\cite{WatanabeEAH21,WatanabeEAH25} introduce string diagrammatic frameworks for parity games and mean-payoff games, respectively.
The compositional algorithm proposed in~\cite{WatanabeEAH25}, applicable to both types of games, relies on enumerating positional strategies for each subsystem.
In contrast,~\cite{DBLP:conf/calco/000325} presents a compositional algorithm for string diagrams of parity games that avoids such enumeration, in the spirit of~\cite{DBLP:conf/tacas/WatanabeVHRJ24}. 
While its worst-case complexity is exponential in the number of exits, our strategy refinement algorithm operates in polynomial time.

\section{Conclusion}\label{sec:conclusion}
In this paper, we have presented two approaches for the verification of almost-sure Büchi objectives in sequentially composed MDPs, expressed using string diagrams. 
The first approach is a bottom-up algorithm which computes a \emph{compositional solution} for each leaf of the string diagram and combines them to obtain a solution of the entire string diagram.
The second approach is an iterative algorithm that closely resembles the classical Büchi algorithm, but reasons locally wherever possible. 
Natural directions for future work are extensions to parity objectives and to quantitative variants of the problem.

\clearpage
\bibliographystyle{plain} 
\bibliography{refs}

\ifwithappendix
\clearpage
\appendix

\clearpage
\section{Omitted Contents in \cref{sec:prelims}}
We sketch a derivation of \cref{lem:buchi} for completeness sake.
Recall \cref{lem:buchi}.
\lembuchi*
\begin{proofs}
We start with the characterization in~\cite[Eq. 3]{DBLP:conf/lics/AlfaroH00} and give a brief explanation.
\begin{equation}\label{eq:alfaro_buchi}
\begin{split}
\nu Y.\, \mu X.\, &F(X,Y) \text{ where }\\
&F(X, Y) \defeq ((\verts \setminus \buchiverts) \cap \mathrm{Apre_1}(Y, X)) \cup (\buchiverts \cap \mathrm{Pre_1}(Y))
\end{split}
\end{equation}
In \cref{eq:alfaro_buchi}, $\mathrm{Apre_1}(Y, X)$ denotes the states $s$ such that there exists an action $a$ such that playing $a$ in $s$ ensures that all successors are in $Y$ and one of the successors is in $X$.
$\mathrm{Pre_1}(Y)$ denotes the set of states $s$ such that there exists an action $a$, such that playing $a$ in $s$ ensures that all successors are in $Y$.lose
To summarize: \cref{eq:alfaro_buchi} states that a state $s$ is in $F(X,Y)$ iff there is an action $a$, such that playing $a$ in $s$ ensures that all successors are in $Y$, additionally requiring that there is a successor in $X$ if $s$ is not in $\buchiverts$.

There are two main differences in the setting of our paper and~\cite{DBLP:conf/lics/AlfaroH00}:
\begin{enumerate}
    \item\label{item:diff_verts} In our setting, the vertices are partitioned into $\verts_1$ and $\verts_P$ vertices, while~\cite{DBLP:conf/lics/AlfaroH00} has only player-1 vertices with a transition function with actions. The fact that $\verts_1$ and $\verts_P$ vertices alternate is similar however.
    \item\label{item:diff_buchi} In our setting, the Büchi vertices are probabilistic vertices while in~\cite{DBLP:conf/lics/AlfaroH00} the Büchi vertices are player-1 vertices. 
\end{enumerate}
First, to resolve difference~\ref{item:diff_verts}, we obtain:
\begin{equation*}
F(X,Y) \defeq \Set*{ v \in \verts_1 \midvert \exists v'\in \Post(v)\colon \Post(v') \subseteq Y \text{ and } v \in \buchiverts \vee v' \in \Pre(X) }
\end{equation*}
Next, to resolve difference~\ref{item:diff_buchi}, we check whether $v' \in \buchiverts$ instead of $v$, obtaining:
\begin{equation*}
F(X,Y) \defeq \Set*{ v \in \verts_1 \midvert \exists v'\in \Post(v)\colon \Post(v') \subseteq Y \text{ and } v' \in \buchiverts \vee v' \in \Pre(X) }
\end{equation*}
\end{proofs}

\section{Omitted Contents in \cref{sec:bottom_up}}
Recall \cref{lem:no_lose_equivalent_def}.
\lemnoloseequivalentdef*
\begin{myproof}
Follows from the expansion of the definitions for no-lose strategies and Büchi acceptance (\cref{def:no_lose,def:buchi}) and applying \cref{lem:tr_equivalence} below.
\small
\begin{align*}
&\forall v \in \Reach_\strat(i)\colon [ \Reach_\strat(v) \cap \ex \ne \emptyset ] \vee [ v \models \buchiobj\buchiverts ]\\ {}\Leftrightarrow{}
&\forall v \in \Reach_\strat(i)\colon [ \Reach_\strat(v) \cap \ex \ne \emptyset ] \vee [ \forall v' \in \Reach_\strat(v)\colon \Reach_\strat(v') \cap \buchiverts \ne \emptyset ]\\ {}\Leftrightarrow{}
&\forall v \in \Reach_\strat(i)\colon \forall v' \in \Reach_\strat(v)\colon \bigl[ \Reach_\strat(v) \cap \ex \ne \emptyset \vee \Reach_\strat(v') \cap \buchiverts \ne \emptyset \bigr]\\
\eqcomment{Apply \cref{lem:tr_equivalence} with $v' \in P \defequiv \Reach_\strat(v) \cap \ex \ne \emptyset \vee \Reach_\strat(v') \cap \buchiverts \ne \emptyset$}
{}\Leftrightarrow{}
&\forall v \in \Reach_\strat(i)\colon [ \Reach_\strat(v) \cap \ex \ne \emptyset \vee \Reach_\strat(v) \cap \buchiverts \ne \emptyset] &&\\ {}\Leftrightarrow{}
&\forall v \in \Reach_\strat(i)\colon [ \Reach_\strat(v) \cap (\ex \cup \buchiverts) \ne \emptyset ]
\end{align*}
\end{myproof}

\begin{lemma}\label{lem:tr_equivalence}
Let $\omdpc A$ be an roMDP, $i \in \verts$ a vertex, $\strat$ a strategy, and $P \subseteq \verts$ some proposition over the vertices, then: 
\begin{align*}
\bigl[\forall v \in \Reach_\strat(i)\colon \forall v' \in \Reach_\strat(v)\colon v' \in P \bigr] \text{ iff }
\bigl[\forall v \in \Reach_\strat(i)\colon v \in P\bigr]
\end{align*}
\end{lemma}
\begin{myproof}
We prove both directions separately.
\begin{description}
\item[($\Rightarrow$):]
Given $x \in \Reach_\strat(i)$ we show that $x \in P$.
If we replace $v$ and $v'$ by $x$ on the left-hand side of the equation, we obtain the required result, where $x \in \Reach_\strat(x)$ holds by the definition of $\Reach$.

\item[($\Leftarrow$):]
Given $v \in \Reach_\strat(i)$ and $v' \in \Reach_\strat(v)$ we show that $v' \in P$.
By \cref{lem:reach_monotone} we obtain $\Reach_\strat(v') \subseteq \Reach_\strat(v) \subseteq \Reach_\strat(i)$, thus, $v' \in \Reach_\strat(i)$ and therefore by applying our assumption that $\forall v \in \Reach_{\strat}(i)$, we obtain $v' \in P$.
\end{description}
\end{myproof}
Recall \cref{lem:local_solution_correct}.
\lemlocalsolutioncorrect*
\begin{myproof}
We prove both directions separately.
\begin{description}
    \item[($\Rightarrow$):]
    $\tuple{\emptyset, \top} \in \solarg{\omdpc A}(i)$ implies (by definition) that there exists a no-lose strategy $\strat$ such that $\CExits{\omdpc A}\strat{i} = \emptyset$ and $\Reach^{\omdpc A}_\strat(i) \cap B \ne \emptyset$.
    $\CExits{\omdpc A}\strat{i} = \emptyset$ implies that for all vertices $v \in \Reach^{\omdpc A}_\strat(i)$ we have $\Reach^{\omdpc A}_\strat(v) \cap (B \cup \emptyset) \ne \emptyset$.
    Hence, $i \models^{\sem{\omdpc A}} \buchiobj\buchiverts$.%
    
    \item[($\Leftarrow$):]
    $i \models^{\sem{\omdpc A}} \buchiobj\buchiverts$ implies (by definition) that there exists a strategy $\strat$ such that for all vertices $v \in \Reach^{\omdpc A}_\strat(i)$ we have that $\Reach^{\omdpc A}_\strat(v) \cap B \ne \emptyset$.
    As $\ex^{\omdpc A} \cap \buchiverts = \emptyset$, we have that $\Reach^{\omdpc A}_\strat(i) \cap \ex^{\omdpc A} = \emptyset$ (as no vertex in $B$ can be reached from $\ex^{\omdpc A}$).
    Hence, $\solstrat^{\omdpc A}(i, \strat) = \tuple{\emptyset, \top}$ so $\tuple{\emptyset, \top} \in \solarg{\omdpc A}(i)$.%
\end{description}
\end{myproof}

\begin{lemma}[$\Reach$ monotonicity]\label{lem:reach_monotone}
Let $\strat$ be some strategy and $v, v' \in \verts$ vertices.
Then: 
$v' \in \Reach_\strat(v) \text{ iff } \Reach_\strat(v') \subseteq \Reach_\strat(v)$.
\end{lemma}
\begin{myproof}
We use the fact that a vertex $v'$ is reachable from $v$ iff there is a finite path from $v$ to $v'$.
We prove both directions separately.
\begin{description}
    \item[($\Rightarrow$):] 
    As $v' \in \Reach_\strat(v)$, there is a finite path $\pi_1$ from $v$ to $v'$.
    To show that $\Reach_\strat(v') \subseteq \Reach_\strat(v)$, we show that $v'' \in \Reach_\strat(v')$ implies $v'' \in \Reach_\strat(v)$.
    As $v'' \in \Reach_\strat(v')$ there is a finite path $\pi_2$ from $v'$ to $v''$.
    We concatenate $\pi_1$ and $\pi_2$ to obtain a path from $v$ to $v''$ and thus, $v'' \in \Reach_\strat(v)$, meaning that $\Reach_\strat(v') \subseteq \Reach_\strat(v)$.
    
    \item[($\Leftarrow$):] 
    By definition of $\Reach$: $v' \in \Reach_\strat(v')$. If $\Reach_\strat(v') \subseteq \Reach_\strat(v)$ then $v' \in \Reach_\strat(v)$.
\end{description}
\end{myproof}

\subsection{Sequential Composition}
Recall \cref{lem:sequential_composition}.
\lemsequentialcomposition*
\begin{myproof}
Let $\strat' \defeq \strat \cup \strat_1 \cup \dots \cup \strat_N$.
We prove the following statements which together prove the lemma.
\begin{itemize}
    \item $\strat'$ is a no-lose strategy.
    
    We have to show that for all $v$ in $\Reach_{\strat'}^{\omdpc A \seqcomp \omdpc B}(i)$ we have that  $\Reach_{\strat'}^{\omdpc A \seqcomp \omdpc B}(v) \cap (\ex^{\omdpc B} \cup \buchiverts) \ne \emptyset$.

    If $v \in \verts^{\omdpc A}$, then by playing $\strat$ we have that $\Reach_\strat^{\omdpc A}(v) \cap (\ex^{\omdpc A} \cup \buchiverts) \ne \emptyset$ as $\strat$ is a no-lose strategy in $\omdpc A$.
    Reaching an exit $x_k \in \ex^{\omdpc A}$ implies by \cref{lem:reach_monotone} and the fact that from $x_k$ onward we follow a sub-strategy of $\strat_k$ which is a no-lose strategy in $\omdpc B$, that indeed $\Reach_{\strat'}^{\omdpc A \seqcomp \omdpc B}(v) \cap (\ex^{\omdpc B} \cup \buchiverts) \ne \emptyset$.
    
    If $v \in \verts^{\omdpc B}$, by definition of $\Reach$ and $\strat'$, necessarily for some strategy $\strat_k$ we have that $v \in \strat_k(v')$ for some predecessor $v' \in \verts^{\omdpc B}$ (note that if the predecessor $v' \in \verts^{\omdpc A}$, then $v=x_k$).
    As $\strat_k$ is a no-lose strategy, we obtain that $\Reach_{\strat_k}^{\omdpc B}(v) \cap (\ex^{\omdpc B} \cup \buchiverts) \ne \emptyset$ and hence $\Reach_{\strat'}^{\omdpc B}(v) \cap (\ex^{\omdpc B} \cup \buchiverts) \ne \emptyset$.
    
    \item $\CExits{\omdpc A \seqcomp \omdpc B}{\strat'}{i} = \bigcup_{1 \le k \le N} \CExits{\omdpc B}{\strat_k}{x_k}$.
    
    Follows directly from \cref{lem:reachable_states_init} below.
    \item $\Reach_{\strat'}^{\omdpc A \seqcomp \omdpc B}(i) \cap B \ne \emptyset \iff b \vee \bigvee_{1 \le k \le N} b_k$.
    
    Follows directly from \cref{lem:reachable_states_init} below.
\end{itemize}

\end{myproof}
Recall \cref{lem:comp_of_seq}.
\lemcompofseq*\noindent
We prove both directions separately in the lemmas below.

\begin{lemma}[$\seqcomp$ solution soundness]%
Let $\omdpc A\colon m \to l$ and $\omdpc B\colon l \to n$ be roMDPs, $i \in [m]$ and $\tuple{T, b} \in \powerset{[n]} \times \{\bot, \top\}$, then
\[
\tuple{T, b} \in (\solarg{\omdpc A} \seqcomp \solarg{\omdpc B})(i) \text{ implies } \tuple{T, b} \in \solarg{\omdpc A \seqcomp \omdpc B}(i).
\]
\end{lemma}
\begin{myproof}
This follows from \cref{lem:sequential_composition}.
\end{myproof}

\begin{lemma}[$\seqcomp$ solution completeness]%
Let $\omdpc A\colon m \to l$ and $\omdpc B\colon l \to n$ be roMDPs, $i \in [m]$ and $\tuple{T, b} \in \powerset{[n]} \times \{\bot, \top\}$, then
\[
\tuple{T, b} \in \solarg{\omdpc A \seqcomp \omdpc B}(i) \text{ implies } \tuple{T, b} \in (\solarg{\omdpc A} \seqcomp \solarg{\omdpc B})(i).
\]
\end{lemma}
\begin{myproof}
Let $\strat'$ be the strategy such that $\solstrat^{\omdpc A \seqcomp \omdpc B}(i, \strat') = \tuple{T, b}$.
Then, we can first split off $\strat\colon \verts^{\omdpc A} \to \powerset{\verts^{\omdpc A}}$ where $\strat(v) \defeq \strat'(v)$.
To be able to use \cref{lem:sequential_composition}, we will have to show that $\strat$ is a no-lose strategy in $\omdpc A$.
In other words, that for all vertices $v \in \Reach_\strat^{\omdpc A}(i)$ we have that $\Reach_\strat^{\omdpc A}(v) \cap (\ex^{\omdpc A} \cup \buchiverts) \ne \emptyset$.
As $\strat'$ is a no-lose strategy in $\omdpc A \seqcomp \omdpc B$, there is some finite path from $v$ to some vertex $v'$ in $\ex^{\omdpc B} \cup \buchiverts$.
If $v' \in \verts^{\omdpc A}$, we can also reach it in $\omdpc A$ by playing $\strat$, otherwise, if $v' \in \verts^{\omdpc B}$, then it must reach some exit in $\omdpc A$.
Therefore, $\strat$ is a no-lose strategy in $\omdpc A$.

Let $\tuple{T_1, b_1} \defeq \solstrat^{\omdpc A}(i, \strat)$.
First, if $T_1 = \emptyset$, then necessarily $b_1 = \top$, as $\strat$ is no-lose.
In this case we immediately obtain that $\tuple{T_1, b_1} = \tuple{T, b}$ and therefore $\tuple{T, b} \in (\solarg{\omdpc A} \seqcomp \solarg{\omdpc B})(i)$.

We continue with the case that $T_1 = \{ x_1, \dots, x_N \} \ne \emptyset$.
Next, we define strategies $\strat_k\colon \verts^{\omdpc B} \to \powerset{\verts^{\omdpc B}}$ for $1 \le k \le N$ such that $\strat_k(v) \defeq \strat'(v)$ if $v \in \Reach_{\strat'}^{\omdpc B}(x_k)$ and $\emptyset$ otherwise.
By definition $\strat_k$ is $x_k$-local.
As $\strat'$ is no-lose in $\omdpc A \seqcomp \omdpc B$, for all vertices $v$ reachable from $x_k$ (which are all reachable states in $\omdpc B$, see \cref{lem:reachable_states_init}), there is a path from $v$ to $v' \in \ex^{\omdpc B} \cup \buchiverts$ in $\omdpc A \seqcomp \omdpc B$.
Because vertices of $\omdpc A$ are unreachable from $\omdpc B$ the path is necessarily completely in $\omdpc B$, and therefore also present in $\omdpc B$ by playing $\strat_k$.
Therefore, $\strat_k$ is a no-lose strategy.

Finally, using \cref{lem:sequential_composition} and the fact that $\strat' = \strat \cup \strat_1 \cup \dots \cup \strat_N$, we obtain
\begin{align*}
    \solstrat^{\omdpc A \seqcomp \omdpc B}(i, \strat') &= \tuple{T, b}\\
    &= \solstrat^{\omdpc A \seqcomp \omdpc B}(i, \strat \cup \strat_1 \cup \dots \cup \strat_N)\\
    &= \tuple{\emptyset, b_1} \sqcup \bigsqcup_{1 \le k \le N} \solstrat^{\omdpc B}(x_k, \strat_k)
\end{align*}
and hence $\tuple{T, b} \in (\solarg{\omdpc A} \seqcomp \solarg{\omdpc B})(i)$ as required.
\end{myproof}

\begin{lemma}\label{lem:reachable_states_init}
Let $\strat$ be a no-lose strategy in $\omdpc A$ such that $\tuple{\{x_1, \dots, x_N\}, b} \defeq \solstrat^{\omdpc A}(i, \strat)$. For $1 \le k \le N$, let $\strat_k$ be a no-lose strategy in $\omdpc B$ and let $\tuple{Y_k, b_k} \defeq \solstrat^{\omdpc B}(x_k, \strat_k)$.
Let $\strat' \defeq \strat \cup \strat_1 \cup \dots \cup \strat_N$.
We have that
\[
\Reach_{\strat'}^{\omdpc A \seqcomp \omdpc B}(i) =
\Reach_\strat^{\omdpc A}(i) \cup \bigcup_{1 \le k \le N} \Reach_{\strat_k}^{\omdpc B}(x_k).
\]
\end{lemma}
\begin{myproof}
We show both inclusions.
\begin{description}
    \item[($\subseteq$):] 
    Given $v \in \Reach_{\strat'}^{\omdpc A \seqcomp \omdpc B}(i)$, we show that $v \in \Reach_\strat^{\omdpc A}(i) \cup \bigcup\limits_{1 \le k \le N} \Reach_{\strat_k}^{\omdpc B}(x_k)$.

    By definition, either $v \in \verts^{\omdpc A}$ or $v \in \verts^{\omdpc B}$.
    If $v \in \verts^{\omdpc A}$, we must have that $v \in \Reach_\strat^{\omdpc A}(i)$, as $\strat$ and $\strat'$ are equal on $\omdpc A$ (and $\omdpc A$ is unreachable from $\omdpc B$).
    
    Otherwise, if $v \in \verts^{\omdpc B}$, then there is some finite path $\pi$ in $\omdpc A \seqcomp \omdpc B$ from $i$ to $v$.
    Let $l \defeq \abs*\pi$.
    Then, by definition of $\Reach$ and $\strat'$, there is some strategy $\strat_k$ such that $v \in \strat_k(\pi_{l-1})$.
    As $\strat_k$ is $x_k$-local, we have that $v \in \Reach_{\strat_k}^{\omdpc B}(x_k)$.
    
    \item[($\supseteq$):]
    Given $v \in \Reach_\strat^{\omdpc A}(i) \cup \bigcup\limits_{1 \le k \le N} \Reach_{\strat_k}^{\omdpc B}(x_k),$ we show that $v \in \Reach_{\strat'}^{\omdpc A \seqcomp \omdpc B}(i)$.
    Again, if $v \in \verts^{\omdpc A}$ then $v \in \Reach_{\strat'}^{\omdpc A \seqcomp \omdpc B}(i)$ follows from the fact that $\strat$ and $\strat'$ are equal on $\omdpc A$.
    Otherwise, if $v \in \verts^{\omdpc B}$ it is reachable by some strategy $\strat_k$ from some entrance $x_k$ (by definition).
    As all $x_k$ are reachable from $i$ it follows from \cref{lem:reach_monotone} that $v \in \Reach_{\strat'}^{\omdpc A \seqcomp \omdpc B}(i)$.
\end{description}

\end{myproof}

\subsection{Trace}
Recall \cref{lem:trace_decomposition}.
\lemtracedecomposition*
\begin{myproof}
Let $\strat \defeq \strat_1 \cup \strat_2 \cup \strat'$.
We have to prove the following statements:
\begin{enumerate}
    \item $\strat$ is a no-lose strategy in $\tr(\omdpc A)$,
    \item $\Reach_\strat^{\tr(\omdpc A)}(i) \cap B \ne \emptyset \iff b_1 \vee b_2$,
    \item $\CExits{\tr(\omdpc A)}{\strat}{i} = (T_1 \cup T_2) \setminus \{\exarg{n+1}^{\omdpc A}\}$.
\end{enumerate}
We prove these statements in the order specified.
\begin{enumerate}
\item We need to show that for all $v$ in $\Reach_\strat^{\tr(\omdpc A)}(i)$ we have that $\Reach_{\strat}^{\tr(\omdpc A)}(v) \cap (\ex^{\tr(\omdpc A)} \cup \buchiverts) \ne \emptyset$.

We distinguish two cases: either (a) $\exarg{n+1}^{\omdpc A} \in \Reach^{\omdpc A}_\strat(v)$ or (b) $\exarg{n+1}^{\omdpc A} \not\in \Reach^{\omdpc A}_\strat(v)$.
\begin{enumerate}
    \item In this case, entrance $m+1$ is reachable through $\exarg{n+1}$, and therefore by the monotonicity of reach (\cref{lem:reach_monotone}), it suffices to show that $\Reach_{\strat_2}^{\omdpc A}(m+1) \cap (\ex^{\tr(\omdpc A)} \cup \buchiverts) \ne \emptyset$.
    By definition, we have that $\tuple{T_2, b_2} \defeq \solstrat^{\omdpc A}(m+1, \strat_2)$ such that $\tuple{T_2, b_2} \ne \tuple{\{\exarg{n+1}^{\omdpc A}\}, \bot}$.
    Thus, either $\ex^{\tr(\omdpc A)}$ or $\buchiverts$ is reachable from $m+1$ by playing $\strat_2$ and therefore from $v$ by playing $\strat$.
    
    \item Due to \cref{lem:strat_union_trace_initial} we have the following (possibly overlapping) two cases: (1) $v \in \Reach_{\strat_1}^{\omdpc A}(i)$, or (2) $v \in \Reach_{\strat_2}^{\omdpc A}(m+1)$.
    If $v \in \Reach_{\strat_1}^{\omdpc A}(i)$, then as $\strat_1$ is a no-lose strategy, we must have that $\Reach_{\strat_1}^{\omdpc A}(v) \cap (\ex^{\tr(\omdpc A)} \cup \buchiverts) \ne \emptyset$.
    
    Otherwise, if $v \in \Reach_{\strat_2}^{\omdpc A}(m+1)$, then as $\strat_2$ is a no-lose strategy such that $\tuple{T_2, b_2} \defeq \solstrat^{\omdpc A}(m+1, \strat_2)$ and $\tuple{T_2, b_2} \ne \tuple{\{\exarg{n+1}^{\omdpc A}\}, \bot}$, we have the following cases:
    \begin{itemize}
        \item $T_2 \supset \Set*{ \exarg{n+1}^{\omdpc A} }$: this implies that we must be able to reach either $\buchiverts$ or $\ex^{\tr(\omdpc A)}$.
        \item $T_2 = \Set*{ \exarg{n+1}^{\omdpc A} }$ and $b_2 = \top$, in which case we can either reach $\buchiverts$ directly or through $\exarg{n+1}^{\omdpc A}$ and then entrance $m+1$ (\cref{lem:reach_monotone}).
        \item $T_2 = \emptyset$ and $b_2 = \top$, this implies that we can reach $\buchiverts$ directly in $\omdpc A$ by playing $\strat_2$, and by extension, by playing $\strat$ in $\tr(\omdpc A)$.
    \end{itemize}
\end{enumerate}

\item The following equivalences hold.
\begin{align*}
    &\Reach^{\tr(\omdpc A)}_{\strat}(i) \cap \buchiverts \ne \emptyset \\
    {}\Leftrightarrow{}& (\Reach^{\omdpc A}_{\strat_1}(i) \cup \Reach^{\omdpc A}_{\strat_2}(m+1)) \cap \buchiverts \ne \emptyset & \text{(\cref{lem:strat_union_trace_initial})}&\\
    {}\Leftrightarrow{}& (\Reach^{\omdpc A}_{\strat_1}(i) \cap \buchiverts \ne \emptyset) \vee (\Reach^{\omdpc A}_{\strat_2}(m+1) \cap \buchiverts \ne \emptyset) \\
    {}\Leftrightarrow{}& b_1 \vee b_2\\
\end{align*}

\item $\CExits{\tr(\omdpc A)}{\strat}{i} = (T_1 \cup T_2) \setminus \Set*{ \exarg{n+1}^{\omdpc A} }$.
The following equalities hold.
\begin{align*}
    &\Reach^{\tr(\omdpc A)}_{\strat}(i) \cap \ex^{\tr(\omdpc A)}\\
    {}={}& \Big(\Reach^{\omdpc A}_{\strat_1}(i) \cup \Reach^{\omdpc A}_{\strat_2}(m+1) \Big) \cap \ex^{\tr(\omdpc A)} & \text{(\cref{lem:strat_union_trace_initial})}&\\
    {}={}& \Big(\Reach^{\omdpc A}_{\strat_1}(i) \cap \ex^{\tr(\omdpc A)}\Big) \cup \Big(\Reach^{\omdpc A}_{\strat_2}(m+1) \cap \ex^{\tr(\omdpc A)}\Big)\\
    {}={}& (T_1 \cup T_2) \setminus \Set*{ \exarg{n+1}^{\omdpc A} } & (\ex^{\tr(\omdpc A)} \defeq \ex^{\omdpc A} \setminus \Set*{ \exarg{n+1}^{\omdpc A} })&
\end{align*}
\end{enumerate}

\end{myproof}
Recall \cref{lem:comp_of_tr}.
\lemcompoftr*\noindent
We prove both directions of \cref{lem:comp_of_tr} separately.

\begin{lemma}[$\tr$ solution soundness]
Let $\omdpc A$ be an roMDP such that\par\noindent
$\typemdp{\omdpc A}\colon m+1 \to n+1$, $i \in [m]$ and $\tuple{T, b} \in \powerset{[n]} \times \{\bot, \top\}$.
Then:
\[ \tuple{T, b} \in (\tr(\solarg{\omdpc A}))(i)\quad\text{implies}\quad\tuple{T, b} \in \solarg{\tr(\omdpc A)}(i).\]
\end{lemma}
\begin{myproof}
We expand the definition of $(\tr(\solarg{\omdpc A}))(i)$ to the following two cases which prove separately.
\begin{enumerate}
    \item $\tuple{T, b} \in \solarg{\omdpc A}(i) \wedge \exarg{n+1} \not\in T$:

    This implies there exists a no-lose strategy $\strat$ such that $\solstrat^{\omdpc A}(i, \strat) = \tuple{T, b}$.
    As the trace exit $\exarg{n+1}$ is not reached, $\strat$ can also be used to show $\solstrat^{\tr(\omdpc A)}(i, \strat) = \tuple{T, b}$, meaning that $\tuple{T, b} \in \solarg{\tr(\omdpc A)}(i)$ as required.
    
    \item There exists $\tuple{T_1, b_1} \in \solarg{\omdpc A}(i)$ and $\tuple{T_2, b_2} \in \solarg{\omdpc A}(m+1)$ such that:
    \begin{align}
    n+1 &\in T_1\label{eq:tr_t_one}\\
    \tuple{T_2, b_2} &\ne \tuple{\{n+1\}, \bot}\label{eq:tr_t_two}\\
    \tuple{(T_1 \cup T_2) \setminus \{n+1\}, b_1 \vee b_2} &= \tuple{T, b}
    \end{align}
The above implies there exists strategies $\strat_1, \strat_2$ such that $\solstrat^{\omdpc A}(i, \strat_1) = \tuple{T_1, b_1}$ and $\solstrat^{\omdpc A}(m+1, \strat_2) = \tuple{T_2, b_2}$.
We construct the strategy $\strat \defeq \strat_1 \cup \strat_2 \cup \strat'$ where $\strat' \defeq \{\exarg{n+1} \mapsto \{*\} \} \cup \{ v \mapsto \emptyset \mid v \in \verts^{\omdpc A}, v \ne\exarg{n+1} \}$.
We satisfy all conditions to apply \cref{lem:trace_decomposition} and obtain
\begin{align*}
\solstrat^{\tr(\omdpc A)}(i, \strat) 
= \tuple{(T_1 \cup T_2) \setminus \{\exarg{n+1}\}, b_1 \vee b_2}
= \tuple{T, b}
\end{align*}
and therefore $\tuple{T, b}\in \solarg{\tr(\omdpc A)}(i)$ as required.
\end{enumerate}
\end{myproof}

\begin{lemma}[$\tr$ solution completeness]
Let $\omdpc A$ be an roMDP such that\par\noindent
$\typemdp{\omdpc A}\colon\allowbreak m+1 \to n+1$, $i \in [m]$ and $\tuple{T, b} \in \powerset{[n]} \times \{\bot, \top\}$.
Then:
\[ \tuple{T, b} \in \solarg{\tr(\omdpc A)}(i)\text{ implies }\tuple{T, b} \in (\tr(\solarg{\omdpc A}))(i).\]
\end{lemma}
\begin{myproof}
$\tuple{T, b} \in \solarg{\tr(\omdpc A)}(i)$ implies (by definition) the existence of a no-lose strategy $\strat$ such that $\solstrat^{\tr(\omdpc A)}\pars*{i, \strat} = \tuple{T, b}$.
We distinguish two cases: (1) $\exarg{n+1}^{\omdpc A} \not\in \Reach_{\strat}^{\tr(\omdpc A)}(i)$, and (2) $\exarg{n+1}^{\omdpc A} \in \Reach_\strat^{\tr(\omdpc A)}(i)$.
\begin{enumerate}
    \item $\exarg{n+1}^{\omdpc A} \not\in \Reach_{\strat}^{\tr(\omdpc A)}(i)$:
    
    By definition of trace, all vertices and transitions are equal except for in the trace exit $\exarg{n+1}^{\omdpc A}$, hence, we have that $\solstrat^{\omdpc A}\pars*{i, \strat} = \solstrat^{\tr(\omdpc A)}\pars*{i, \strat} = \tuple{T, b}$.
    This implies that $\tuple{T, b} \in (\tr(\solarg{\omdpc A}))(i)$ as required.

    \item $\exarg{n+1}^{\omdpc A} \in \Reach_\strat^{\tr(\omdpc A)}\pars*{i}$:

We split $\strat$ into $\strat_1, \strat_2\colon \verts^{\omdpc A} \to \powerset{\verts^{\omdpc A}}$ where
{\small
\begin{align*}
\strat_1(v) \defeq \begin{cases}
    \strat(v) & \text{if } v \in \Reach_\strat^{\omdpc A}(i)\\
    \emptyset & \text{otherwise}
\end{cases},\quad
\strat_2(v) \defeq \begin{cases}
    \strat(v) & \text{if } v \in \Reach_\strat^{\omdpc A}(m+1)\\
    \emptyset & \text{otherwise}
\end{cases}.
\end{align*}
}
Note that by definition, $\strat_1$ and $\strat_2$ are $i$-local and $(m+1)$-local respectively.
We work towards satisfying the conditions of \cref{lem:trace_decomposition}.
We require:
\begin{itemize}
    \item $\omdpc A$ is an roMDP such that $\typemdp{\omdpc A}\colon m+1 \to n+1$ and $i \in [m]$.
    
    This is satisfied by definition.
    
    \item $\strat_1, \strat_2\colon \verts^{\omdpc A} \to \powerset{\verts^{\omdpc A}}$ are no-lose strategies in $\omdpc A$ for entrances $i$ and $m+1$ respectively.

    \begin{itemize}
    \item To show that $\strat_1$ is a no-lose strategy in $\omdpc A$, we need to show that for all $v$ in $\Reach_{\strat_1}^{\omdpc A}(i)$ we have that $\Reach_{\strat_1}^{\omdpc A}(v) \cap (\ex^{\omdpc A} \cup \buchiverts) \ne \emptyset$.
    Consider the shortest path $\pi$ from $v$ to $v' \in (\ex^{\tr(\omdpc A)} \cup \buchiverts)$ in $\tr(\omdpc A)$ which necessarily exists as $\strat$ is a no-lose strategy in $\tr(\omdpc A)$.
    Either $\pi$ does not contain $\exarg{n+1}$, in which case $\pi$ also exists in $\omdpc A$ by playing $\strat_1$ (which is equal to $\strat$ on the reachable states), or otherwise, we reach $\exarg{n+1}$ and therefore $\ex^{\omdpc A}$.
    
    \item To show that $\strat_2$ is a no-lose strategy in $\omdpc A$ we use the same argument as for $\strat_1$.
    \end{itemize}

    \item $\tuple{T_1, b_1} \defeq \solstrat^{\omdpc A}(i, \strat_1)$ such that $\exarg{n+1}^{\omdpc A} \in T_1$.
    
    Follows from the fact that $\exarg{n+1}^{\omdpc A} \in \Reach_\strat^{\tr(\omdpc A)}(i)$ and therefore $\exarg{n+1}^{\omdpc A} \in \Reach_{\strat_1}^{\omdpc A}(i)$.
    
    \item $\tuple{T_2, b_2} \defeq \solstrat^{\omdpc A}(m+1, \strat_2)$ and $\tuple{T_2, b_2} \ne \tuple*{\Set*{\exarg{n+1}^{\omdpc A}}, \bot}$.
    
    Assume for contradiction the opposite, namely that $\tuple{T_2, b_2} = \tuple{\{\exarg{n+1}^{\omdpc A}\}, \bot}$.
    This would imply that for entrance $m+1$ in $\tr(\omdpc A)$ by playing $\strat_2$ (and therefore $\strat$) we cannot reach a Büchi vertex or an exit in $\ex^{\tr(\omdpc A)}$.
    This implies that $\strat$ is losing in $\tr(\omdpc A)$, which is a contradiction.

    \item $\strat' \defeq \{ \exarg{n+1} \mapsto \{*\} \} \cup \{ v \mapsto \emptyset \mid v \in \verts^{\omdpc A}, v \ne\exarg{n+1} \}$:
    
    Satisfied by definition.

\end{itemize}
We have that $\strat = \strat_1 \cup \strat_2 \cup \strat'$, which 
follows by definition of $\strat_1$, $\strat_2$, $\strat'$ and \cref{lem:strat_union_trace_initial}.
We satisfy all conditions to apply \cref{lem:trace_decomposition} and obtain
\begin{align*}
\solstrat^{\tr(\omdpc A)}(i, \strat) &= \tuple{T, b}\\
&= \solstrat^{\tr(\omdpc A)}(i, \strat_1 \cup \strat_2 \cup \strat')\\
&= \tuple*{(T_1 \cup T_2) \setminus \Set*{ \exarg{n+1}^{\omdpc A} }, b_1 \vee b_2}
\end{align*}
and hence $\tuple{T, b} \in (\tr(\solarg{\omdpc A}))(i)$.
\end{enumerate}
\end{myproof}

\begin{lemma}[$\Reach$ decomposition $\tr$]\label{lem:strat_union_trace_initial}
Let
\begin{itemize}
    \item $\omdpc A$ be an roMDP such that $\typemdp{\omdpc A}\colon m+1 \to n+1$ and $i \in [m]$,
    \item $\strat_1, \strat_2\colon \verts^{\omdpc A} \to \powerset{\verts^{\omdpc A}}$ be no-lose strategies in $\omdpc A$ for $i$ and $m+1$ respectively,
    \item $\tuple{T_1, b_1} \defeq \solstrat^{\omdpc A}\pars*{i, \strat_1}$ such that $\exarg{n+1}^{\omdpc A} \in T_1$,
    \item $\tuple{T_2, b_2} \defeq \solstrat^{\omdpc A}\pars*{m+1, \strat_2}$ such that $\tuple{T_2, b_2} \ne \tuple{\{\exarg{n+1}^{\omdpc A}\}, \bot}$,
    \item $\strat' \defeq \{\exarg{n+1} \mapsto \{*\} \} \cup \{ v \mapsto \emptyset \mid v \in \verts^{\omdpc A}, v \ne\exarg{n+1} \}$,
    \item $\strat \defeq \strat_1 \cup \strat_2 \cup \strat'$.
\end{itemize}
Then,
$
    \Reach^{\tr(\omdpc A)}_{\strat_1 \cup \strat_2 \cup \strat'}(i) =
    \Reach^{\omdpc A}_{\strat_1}(i) \cup \Reach^{\omdpc A}_{\strat_2}(m+1).
$
\end{lemma}
\begin{myproof}
We prove both inclusions.
\begin{description}
    \item[($\supseteq$):] 
    This follows from the fact that entrance $m+1$ is reachable through $\exarg{n+1}$ combined with the fact that $\strat(v) \supseteq \strat_1(v)$ and $\strat(v) \supseteq \strat_2(v)$ for all $v \in \verts^{\omdpc A}$.

    \item[$(\subseteq)$:]
    This follows from the definition of $\Reach$ and the fact that $\strat_1$ and $\strat_2$ are $i$ and $(m+1)$-local respectively.
    Given vertex $v \in \Reach^{\tr(\omdpc A)}_{\strat_1 \cup \strat_2 \cup \strat'}(i)$, then there is some predecessor $v'$ such that either $v \in \strat_1(v')$ or $v \in \strat_2(v')$, which immediately implies that either $v \in \Reach^{\omdpc A}_{\strat_1}(i)$ or $\Reach^{\omdpc A}_{\strat_2}(m+1)$.
\end{description}
\end{myproof}

\subsection{One-shot Bottom-up Algorithm}\label{sec:one_shot_bottom_up_algorithm}
We present our one-shot bottom-up algorithm with solution sharing in \cref{alg:bottomup}.
\begin{algorithm}[t]
\newcommand{\cache}{\mathrm{Cache}}
\newcommand{\algsol}{\textsc{Sol}}
\caption{One-shot Algorithm}
    \begin{algorithmic}[1]
        \Procedure{OneShot}{$\sd$, $\buchiverts \subseteq \verts^{\sem\sd}$, $i \in \en^{\sem\sd}$}
            \State \Return $\tuple{\emptyset, \top} \in \bigl(\algsol(\sd, B)\bigr)(i)$
        \EndProcedure
\item[]
        \State $\cache \assign \emptyset$
        \Procedure{Sol}{$\sd$, $\buchiverts$}
            \If{$\sd \not\in \cache$}
                \If{$\sd = \omdpc A$}
                    $\cache(\sd) \assign \solarg{\omdpc A}$
                \ElsIf{$\sd = \sd_1 \seqcomp \sd_2$}
                    $\cache(\sd) \assign \algsol(\sd_1, \buchiverts) \seqcomp \algsol(\sd_2, \buchiverts)$
                \ElsIf{$\sd = \sd_1 \sumcomp \sd_2$}
                    $\cache(\sd) \assign \algsol(\sd_1, \buchiverts) \sumcomp \algsol(\sd_2, \buchiverts)$
                \ElsIf{$\sd = \tr(\sd_1)$}
                    $\cache(\sd) \assign \tr(\algsol(\sd_1, \buchiverts))$
                \EndIf
            \EndIf
            \State \Return $\cache(\sd)$
        \EndProcedure
    \end{algorithmic}
    \label{alg:bottomup}
\end{algorithm}
We provide a small example on how solution sharing can increase the efficiency of the algorithm.
\begin{example}[solution sharing]
    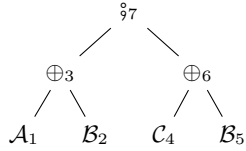
\begin{figure}[t]
        \centering
        \begin{tikzpicture}[node distance=0.5]
            \node (root) {$\seqcomp_7$};
            \node[below left=of root] (a) {$\sumcomp_3$};
            \node[below right=of root] (b) {$\sumcomp_6$};
            \node[below left=of a, xshift=0.5cm] (a1) {$\omdpc A_1$};
            \node[below right=of a, xshift=-0.5cm] (a2) {$\omdpc B_2$};
            \node[below left=of b, xshift=0.5cm] (b1) {$\omdpc C_4$};
            \node[below right=of b, xshift=-0.5cm] (b2) {$\omdpc B_5$};
        
            \draw
                (root) -- (a) -- (a1) (a) -- (a2)
                (root) -- (b) -- (b1) (b) -- (b2);
        \end{tikzpicture}
        \caption{Solution sharing example}
        \label{fig:example_sharing}
    \end{figure}
    Consider the string diagram $\sd \defeq (\omdpc A \sumcomp \omdpc B) \seqcomp (\omdpc C \sumcomp \omdpc B)$ depicted in \cref{fig:example_sharing}.
    In this example, we assume that $\solarg{\omdpc A} = \solarg{\omdpc C}$.
    We use the numbers in \cref{fig:example_sharing} to explain the execution order of our algorithm.
    \begin{enumerate}
        \item[(1,2):] We recursively traverse the string diagram and compute $\solarg{\omdpc A}$ and  $\solarg{\omdpc B}$.
        \item[(3):] We combine the results to obtain $\solarg{\omdpc A} \sumcomp \solarg{\omdpc B} = \solarg{\omdpc A \sumcomp \omdpc B}$.
        \item[(4):] We recursively traverse the other branch of the sequential composition and compute $\solarg{\omdpc C}$ and notice that it is equal to the solution computed in (1).
        \item[(5):] We compute $\solarg{\omdpc B}$ by reusing the solution computed in (1).
        \item[(6):] We combine the results of (4) and (5) which must be equal to the sum of (1) and (2), which was computed in (2), and can therefore obtain the result from the cache.
        \item[(7):] We combine the results of (3) and (6).
    \end{enumerate}
    Without solution sharing, we compute the solution of 4 leaf roMDPs and 3 compositional operations.
    By sharing, we compute the solution of 3 leafs and performed 2 composition operations instead.
    In general, sharing bigger subtrees improves the effectiveness of the cache.
\end{example}

\section{Omitted Contents of \cref{sec:refinement}}

\subsection{Proofs of \cref{sec:shortcut_graph}}\label{app:shortcut}
Recall \cref{lem:shortcut_graph}:
\lemshortcutgraph*\noindent
We prove both directions separately.

\begin{lemma}[soundness of $\graph_\sd$]
Let $i \in \componentsent(\sd)$ be a vertex.
We have that
\[
i \models^{\graph_\sd} \buchiobj\buchiverts_\sd \text{ implies }
i \models^{\sem\sd} \buchiobj\buchiverts
\]
\end{lemma}
\begin{myproof}
Let $\strat$ be a strategy such that $i, \strat \models^{\graph_\sd} \buchiobj\buchiverts_\sd$, which exists by definition.
Using $\strat$, we will construct a global strategy $\strat_G$ such that $i, \strat_G \models^{\sem\sd} \buchiobj\buchiverts$.
First, we define local strategies $\strat_t$ for each transition $t \defeq \tuple{i, \tuple{T, b}}$ picked by $\strat$ (i.e., $t \in \strat(i)$) such that $\solstrat^{\omdpc A}\pars*{i, \strat_t} = \tuple{T, b}$.

We sketch the construction of $\strat_G$.
For the initial vertex $i$, we randomly play a strategy according to the transitions of in $\strat(i)$, say $\strat_t$.
As $\strat_t$ is necessarily a no-lose strategy in the roMDP containing entrance $i$, we will either (1) satisfy the Büchi condition locally, or (2) reach an exit $j$, possibly reaching a Büchi vertex along the way.
We perform the same strategy selection for the entrance connected to $j$ as for $i$.
As $\strat$ is winning in $\shortcut$, this implies that we will eventually see a transition such that the Büchi bit is set, and hence play a local strategy in some component such that we can reach a Büchi vertex in that component.
Therefore, playing $\strat_G$ is winning for $i$ in $\sem\sd$.

The strategy $\strat_G$ requires memory to keep track of the entrance that was taken in the current roMDP.
As memoryless strategies are sufficient, there must be a memoryless strategy that is also winning. %
Hence, we claim that $i \models^{\sem\sd} \buchiobj\buchiverts$.
\end{myproof}

\begin{lemma}[completeness of $\graph_\sd$]
Let $i \in \componentsent(\sd)$ be a vertex.
We have that
\[
i \models^{\sem\sd} \buchiobj\buchiverts \text{ implies }
i \models^{\graph_\sd} \buchiobj\buchiverts_\sd.
\]
\end{lemma}
\begin{myproof}
Let $\strat$ be a strategy such that $i, \strat \models^{\sem\sd} \buchiobj\buchiverts$, which exists by definition.
Next, we can split $\strat$ into local strategies based on each entrance:
Let
\begin{align*}
\strat_{\enarg j^{\omdpc A}}(v) \defeq \begin{cases}
    \strat(v) & \text{if } v \in \Reach_\strat^{\omdpc A}(j)\\
    \emptyset & \text{otherwise}
\end{cases}.
\end{align*}
Then every entrance $\enarg j^{\omdpc A}$ that is reachable under $\strat$ when starting in $\enarg j^{\omdpc A}$ has a strategy $\strat_{\enarg j^{\omdpc A}}$.
$\strat_{\enarg j^{\omdpc A}}$ is a no-lose strategy: if $\strat_{\enarg j^{\omdpc A}}$ were a losing strategy, it would imply that $\strat$ is losing, which is a contradiction.

Next, we construct a strategy $\strat'$ in $\graph_\sd$:
\begin{align*}
\strat'\pars*{\enarg j^{\omdpc A}} \defeq \Set*{ \tuple*{\enarg j^{\omdpc A}, \solstrat^{\omdpc A}(\enarg j^{\omdpc A}, \strat_{\enarg j^{\omdpc A}})} } \text{ for } \enarg j^{\omdpc A} \in \Reach_\strat^{\sem\sd}(i) \cap \componentsent(\sd).
\end{align*}
We claim that $i, \strat' \models^{\graph_\sd} \buchiobj\buchiverts_\sd$.
We have to show that for all $v$ in $\Reach_{\strat'}^{\graph_\sd}(i)$ we have that $\Reach_{\strat'}^{\graph_\sd}(v) \cap \buchiverts_\sd \ne \emptyset$.
By definition of $\strat'$, we have that $v$ in $\Reach_{\strat'}^{\graph_\sd}(i)$ implies that $v$ is in $\Reach_\strat^{\sem\sd}(i)$.
As $\strat$ is winning, this implies there is a shortest path $\Path \defeq \Path_1\ \Path_2\ \dots\ \Path_n \in \verts^*$ from $v$ to $b \in \buchiverts$ in $\sem\sd$.
Let $\mathit{Verts}(\Path) \defeq \Set*{ \pi_k \mid k \in [n] }$ denote the set of vertices appearing in $\Path$.
If $\mathit{Verts}(\Path) \subseteq \verts^{\omdpc A}$ then, $\Reach^{\omdpc A}_\strat(v) \cap B \ne \emptyset$ and thus $\Reach_{\strat'}^{\graph_\sd}(v) \cap \buchiverts_\sd \ne \emptyset$ as required.
Otherwise, if $\mathit{Verts}(\Path) \not\subseteq \verts^{\omdpc A}$, then let $\omdpc A_1, \dots, \omdpc A_k$ be the roMDPs `containing' $\Path$ and let $v_1, \dots, v_k \defeq \mathit{Verts}(\Path) \cap \componentsent(\sd)$ be the entrance vertices visited in $\Path$.
By definition of the shortcut graph, it follows that $v_1 = v$ and for all $1 \le l < k$ we have that $v_{l+1} \in T_l$ where $\tuple{T_l, \cdot} = \solstrat(v_l, \strat_{v_l})$ and thus $\tuple{v_l, v_{l+1}} \in \edges^{\graph_\sd}$ and hence $v_k \in \Reach_{\strat'}^{\graph_\sd}(v)$.
Finally, it is clear that $v_k$ is contained within $\verts^{\omdpc A_k}$, and hence we can apply the same logic as before to obtain that $\buchiverts_\sd$ is reached.
\end{myproof}

\subsection{Proofs of \cref{sec:refinement_algorithm}}
Recall \cref{lem:strategy_union}:
\lemstrategyunion*
\begin{proof}
Let $\strat \defeq \strat_1 \cup \strat_2$. %
We prove the following statements which together prove the lemma:
\begin{enumerate}
    \item $\strat$ is a no-lose strategy, i.e., for all vertices $v$ in $\Reach_\strat(i)$ we have that $\Reach_\strat(v) \cap (\ex \cup \buchiverts) \ne \emptyset$.

    According to \cref{lem:no_lose_strategy_union_reachable_states} below it suffices to show that for all $v$ in $\Reach_\strat(i)$ we have that $(\Reach_{\strat_1}(v) \cup \Reach_{\strat_2}(v)) \cap (\ex \cup \buchiverts) \ne \emptyset$.
    This holds as $\strat_1$ and $\strat_2$ are no-lose strategies and all vertices are either reachable by $\strat_1$ or $\strat_2$ (\cref{lem:no_lose_strategy_union_reachable_states_init}).
    
    \item $\Exits\strat i = \Exits{\strat_1}i \cup \Exits{\strat_2}i$.
    Follows immediately from \cref{lem:no_lose_strategy_union_reachable_states_init}.
    \item $\Reach_\strat(i) \cap B \ne \emptyset \iff (\Reach_{\strat_1}(i) \cup \Reach_{\strat_2}(i)) \cap B \ne \emptyset$.
    Follows immediately from \cref{lem:no_lose_strategy_union_reachable_states_init}.
\end{enumerate}
\end{proof}

\begin{lemma}%
\label{lem:no_lose_strategy_union_reachable_states}
Let $\strat_1$ and $\strat_2$ be no-lose strategies from entrance $i$. %
Then, for all $v$ in $\Reach_{\strat_1 \cup \strat_2}(i)$ we have that
$\Reach_{\strat_1 \cup \strat_2}(v) \supseteq \Reach_{\strat_1}(v) \cup \Reach_{\strat_2}(v)$.
\end{lemma}
\begin{myproof}
    Let $\strat \defeq \strat_1 \cup \strat_2$.
    For all $v$ we have that $\strat_1(v) \subseteq \strat(v)$ and $\strat_2(v) \subseteq \strat(v)$ and hence $\Reach_{\strat_1}(v) \subseteq \Reach_\strat(v)$ and $\Reach_{\strat_2}(v) \subseteq \Reach_\strat(v)$. %
\end{myproof}

\begin{lemma}%
\label{lem:no_lose_strategy_union_reachable_states_init}
Let $\strat_1$ and $\strat_2$ be no-lose strategies for entrance $i$. %
Then:
\[
\Reach_{\strat_1 \cup \strat_2}(i) = \Reach_{\strat_1}(i) \cup \Reach_{\strat_2}(i).
\]
\end{lemma}
\begin{myproof}
Let $\strat \defeq \strat_1 \cup \strat_2$.
We prove both inclusions.

$\Reach_\strat(i) \supseteq \Reach_{\strat_1}(i) \cup \Reach_{\strat_2}(i)$ follows immediately from \cref{lem:no_lose_strategy_union_reachable_states}.

$\Reach_\strat(i) \subseteq \Reach_{\strat_1}(i) \cup \Reach_{\strat_2}(i)$ follows from the fact that $\strat_1$ and $\strat_2$ are $i$-local and hence all vertices reachable under $\strat$ are reachable under either $\strat_1$ or $\strat_2$.
\end{myproof}

\clearpage
\subsection{Strategy Refinement Algorithm}\label{sec:app_refinement_alg}
We give pseudocode for the refinement algorithm (without caching) in \cref{alg:refinement}.
The algorithm uses \emph{configurations} $\conf C$ and $\conf C'$, which map component entrances $\enarg i^{\omdpc A}$ to a (single) effect.
Initially, such configurations map each component entrance $\enarg i^{\omdpc A}$ to the biggest effect $\solarg{\omdpc A}\pars*{i, \ex^{\omdpc A}}$.
\begin{algorithm}[H]
\caption{Strategy Refinement Algorithm}
    \begin{algorithmic}[1]
        \Procedure{\refinementalgtext}{$\sd$, $B \subseteq \verts^{\sem\sd}$, $i \in \en^{\sem\sd}$}
            \For{$\enarg i^{\omdpc A} \in \componentsent(\sd)$}
                \State $\conf C'\pars*{\enarg i^{\omdpc A}} \assign \solarg{\omdpc A}\pars*{i, \ex^{\omdpc A}}$
            \EndFor
            \item[]
            
            \Repeat
                \State $\conf C \assign \conf C'$
                \State $\graph_\conf C \assign \graph_\sd \cap \conf C$ \Comment{subgraph of $\graph_\sd$ containing only vertices of $\conf C$}
                \State $\buchiverts_{\conf C} \assign \buchiverts_\sd \cap \conf C$ \Comment{subset of $\buchiverts_\sd$ containing only vertices of $\conf C$}
                \State $Y \assign \canreach(\graph_{\conf C}, \buchiverts_{\conf C})$

                \For{$\enarg i^{\omdpc A} \in Y$}
                    \State $\conf C'\pars*{\enarg i^{\omdpc A}} \assign \solarg{\omdpc A}\pars*{i, \ex^{\omdpc A} \cap Y}$
                \EndFor
            \Until{$\conf C = \conf C'$}

            \State \Return $i \in Y$
        \EndProcedure
\item[]
        \Procedure{\canreachtext}{$\graph$, $\buchiverts$}
            \State $X' \assign \emptyset$
            \Repeat
                \State $X \assign X'$
                \State $X' \assign X \cup \Pre^\graph(B \cup X)$
            \Until{$X = X'$}
            
            \State \Return $X$
        \EndProcedure
    \end{algorithmic}
    \label{alg:refinement}
\end{algorithm}

\subsubsection{Proofs of Caching.}
Recall \cref{lem:biggest_solution_monotone}:
\lembiggestsolutionmonotone*
\begin{myproof}
    In general, for join semilattices, we have that if $E' \subseteq E$ then $\max E' \le \max E$.
    The lemma follows from the fact that the effects of a given entrance form a join semilattice (\cref{lem:lattice}).
\end{myproof}
Recall \cref{lem:solution_reuse}:
\lemsolutionreuse*
\begin{myproof}
We prove the more general statement which proves the lemma:
If $\max X = n$, and $Y \subseteq X$ and $n \in Y$, then $\max Y = n$.
As $Y \subseteq X$, we have that $\max Y \le \max X$.
Thus, $\max Y \le n$, therefore, as $n \in Y$: $\max Y = n$.
\end{myproof}

\fi

\ifwithnotes
\input{notes}
\fi

\end{document}